\documentclass[letterpaper, 10 pt, journal, final]{IEEEtran}
\usepackage{graphicx} 
\usepackage{url}
\usepackage{svg}
\usepackage{graphics} 
\usepackage{epsfig} 
\usepackage{amsmath} 
\usepackage{amsthm}
\usepackage{amssymb}  
\usepackage{dsfont}
\usepackage{subcaption}
\usepackage{algorithmic}
\usepackage[ruled,vlined]{algorithm2e}
\usepackage{multirow}
\usepackage{makecell}
\usepackage{cite}

\usepackage{booktabs}
\usepackage{array}
\usepackage{caption}
\usepackage{float}
\usepackage{xurl}

\newtheorem{definition}{Definition}

\newtheorem{remark}{Remark}

\newtheorem{lemma}{Lemma}
\newtheorem{theorem}{Theorem}

\newtheorem{construction}{Construction}

\newcommand{\Gcal}{\mathcal{G}}
\newcommand{\Ghat}{\hat{\mathcal{G}}}
\newcommand{\Vcal}{\mathcal{V}}
\newcommand{\Ecal}{\mathcal{E}}

\newcommand{\Ehat}{\hat{E}}

\newcommand{\Kcal}{\mathcal{K}}

\newcommand{\nK}{n_{\mathcal{K}}}
\newcommand{\nD}{n_{\mathcal{D}}}

\title{A Framework for Motion Planning with Temporal Logic Precedence Specifications via Augmented Graphs of Convex Sets}
\author{Shilin You$^*$, Gael Luna$^*$, Tyler H. Summers

\thanks{$^*$Equal contribution. The authors are with the University of Texas at Dallas. Email: \texttt{tyler.summers@utdallas.edu}. This work is supported by the United States Air Force Office of Scientific Research under Grant FA9550-23-1-0424, and by the National Science Foundation under Grant ECCS-2047040.}
}

\begin{document}

\maketitle

\begin{abstract}
We present a framework for planning trajectories that avoid obstacles and satisfy logical precedence constraints expressed with a fragment of signal temporal logic (STL). Our approach models environments containing obstacles, \emph{keys}, and \emph{doors}, where collecting a key unlocks its associated door and potentially opens shorter paths to a goal. Based on an exact convex partitioning of the free space that encodes connectivity among convex free space, key, and door regions, we construct an \emph{augmented graph of convex sets} (GCS) whose layered structure exactly encodes the key-door precedence logic. A shortest path in the augmented GCS simultaneously selects an optimal key collection sequence and computes an optimal continuous trajectory, providing an exact solution up to a finite B\'ezier curve parameterization. 

The present work substantially extends \cite{you2025motion} with three main contributions. First, we formally prove soundness, completeness, and cost equivalence of the augmented GCS construction. Second, we establish a precise correspondence between the augmented GCS layer structure and the Bellman-Held-Karp (BHK) dynamic programming algorithm for the Traveling Salesman Problem (TSP), showing that our augmented GCS framework can be interpreted as a continuous-geometry generalization of both the TSP and the shortest Hamiltonian path problem within the same combinatorial complexity class as BHK. Third, we develop an initial library of eight STL mission specification variations with proven correctness, including mandatory key collection, pure wayset collection (without precedence constraints), ordered key collection, disjunctive keys, conjunctive doors, items/tools (multi-door keys), timed keys and doors, and conditional (if-then) constraints. Numerical experiments on an expanded benchmark suite demonstrate that our approach achieves an \emph{exponential} speedup over methods that use general-purpose temporal logic tools, scales to large environments, and yields solutions within 2\% of the global optimum on almost all tested instances.
\end{abstract}

\section{Introduction}
Motion planning is a fundamental problem in robotics and autonomous systems involving the computation of feasible, optimized trajectories for agents navigating through complex environments \cite{lavalle2006planning,latombe2012robot}. Beyond obstacle avoidance, many real-world tasks impose rich logical and sequential constraints on how an agent must interact with its environment. A warehouse robot must pick and verify an item at one station before proceeding to delivery locations. An autonomous underwater robot must make a surface communications handshake before descending into a deep zone. A surgical robot must exchange the correct tool at a swap station before entering the next anatomical tissue layer. An inspection robot must confirm sensor access before traversing a restricted zone. In general, these are \emph{precedence constraints}: conditions under which parts of an environment become accessible only after prior actions have been completed, such as task execution, confirmation, clearance, or manipulation. 

Encoding and solving motion planning problems with such constraints efficiently and optimally is a major challenge, due to the inherent coupling between combinatorial sequencing (determining the order of actions) and continuous geometric optimization (computing a smooth trajectory through the environment). These problems are coupled tightly: the optimal action sequence depends on the geometry, and the feasibility of a continuous path depends on the action sequence. Existing approaches, discussed in detail in the related work section below, tend to handle one problem well, but rarely both. 

The \emph{graphs of convex sets} (GCS) framework, introduced by Marcucci et al. \cite{marcucci2024shortest, marcucci2023motion, marcucci2025unified}, provides a powerful framework to address this discrete-continuous coupling. By associating convex programs with the nodes and edges of a graph, GCS elegantly unifies combinatorial graph search with continuous convex optimization into a single optimization problem. For shortest path problems, solving the relaxation of a tight mixed-integer convex reformulation yields certifiably near-optimal or globally optimal solutions for many practical motion planning problems, essentially using only convex optimization. 

The goal of the present paper is to extend the GCS framework to a much richer class of \emph{mission} planning problems with precedence constraints, expressed as a fragment of signal temporal logic (STL). An overview of the proposed approach is shown in Fig.~\ref{fig:workflow}. Prior work of You et al.~\cite{you2025motion} introduced \emph{augmented graphs of convex sets} that encode key-door precedence constraints within the GCS framework. The main insight was that by constructing carefully layered and modified copies of the base GCS -- one subgraph per reachable subset of collected keys -- and linking layers via zero-cost directed edges at key nodes, the combinatorial precedence structure is exactly encoded in the graph topology. A shortest path in the augmented GCS simultaneously selects an optimal key-collection sequence and computes an optimal continuous trajectory, solving the discrete-continuous problems jointly. On benchmark instances, our approach runs several orders of magnitude faster and scales to far larger environments than general-purpose temporal logic motion planning tools \cite{kurtz2023temporal}.
\begin{figure*}[t]
  \centering
  \includegraphics[width=0.9\textwidth]{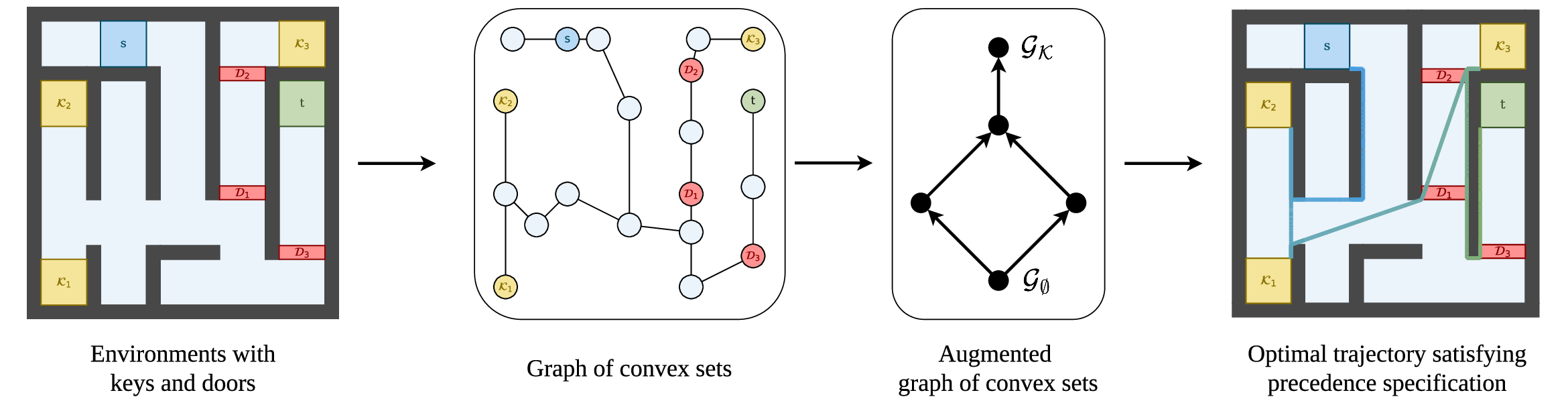}
  \caption{Overview of the proposed approach. The environment is first partitioned and represented as a Graph of Convex Sets (GCS). An augmented GCS is then constructed to encode precedence specifications, and a shortest path problem is solved for the augmented GCS.} 
  \label{fig:workflow}
\end{figure*}

The present paper substantially extends You et al.~\cite{you2025motion}. We provide formal proofs of optimality, establish a precise connection to a landmark classical combinatorial algorithm, develop a library of specification variations, and present an expanded computational study.

\subsection{Contributions}
Prior work \cite{you2025motion} made the following contributions: (1) an exact convex partitioning algorithm (Algorithm 1) that constructs a labeled GCS encoding connectivity among free space, key, and door regions; (2) an algorithm (Algorithm 2) for constructing an augmented GCS with an informal argument that a shortest path in the augmented GCS exactly solves the original motion planning problem; and (3) a key-door benchmark generator and numerical experiments demonstrating orders-of-magnitude speedup over general-purpose temporal logic tools.

The present paper contributes the following beyond \cite{you2025motion}:

\textbf{Formal proofs of correctness and optimality (Theorem 1).} We present proofs of soundness and completeness, establishing a bijection between augmented GCS paths and feasible trajectories, and of cost equivalence.

\textbf{Bellman-Held-Karp (BHK) correspondence (Theorem 2).} We demonstrate that the layer structure and inter-layer edges of the augmented GCS correspond exactly to the BHK dynamic programming recurrence \cite{bellman1962dynamic,held1962dynamic} for the Traveling Salesman Problem (TSP), generalized from discrete city-to-city costs to shortest paths through continuous geometry. The augmented GCS matches the combinatorial complexity class of BHK while additionally solving the continuous geometry problem that BHK takes as preprocessed input.

The BHK correspondence is not just an elegant aesthetic observation, it also opens the door to connect with more than 70 years of research on the TSP, one of the most widely studied combinatorial optimization problems. It strongly motivates translating a wide variety of heuristics into the augmented GCS framework; these heuristics can enable scaling to extremely large instances despite worst-case (exponential) computational complexity.

\textbf{Extended Specification Library (Section \ref{sec:variations}).} We develop eight specification variations extending the key-door precedence logic: mandatory key collection, pure wayset collection (without precedence constraints), ordered key collection, disjunctive keys, conjunctive doors (multi-key doors), items/tools (multi-door keys), timed keys and doors, and conditional (if-then) constraints. For each variation, we provide the modified STL formula, the corresponding modification for augmented GCS construction in Algorithm 2, and a proof of correctness. This extensible library of STL specifications enables moving from vanilla motion planning to increasingly complex \emph{mission} planning and design problems. 

\textbf{Expanded experimental evaluation and refactored codebase.} We present substantially expanded numerical experiments including more varied and larger key-door maze instances, a new benchmark suite for the pure wayset collection (TSP) variation, a hand-designed environment from a classic video game, and a relaxation tightness analysis. For the final version, we will provide a clean refactoring of the codebase to facilitate reproducibility and follow-up research.

\textbf{Improved Exact Convex Partitioning.} Finally, we provide a significant improvement over the convex partitioning algorithm in \cite{you2025motion}. In contrast to obstacles and doors, key regions do not constrain traversability: an agent may pass through a key region freely, and key regions are permitted to overlap with free-space cells. Therefore, only obstacles and doors must contribute hyperplanes to the arrangement that partitions the environment. Key regions are incorporated as labeled nodes after the partition is constructed, by identifying which cells each key region intersects. Excluding key region boundaries from the hyperplane arrangement avoids introducing cuts into the free-space partition that serve no geometric purpose, directly reducing the size of every subgraph in the augmented GCS.

\subsection{Related Work}
\textbf{Motion Planning.} Motion planning comprises sampling-based, combinatorial,
and optimization-based paradigms~\cite{lavalle2006planning, latombe2012robot}.
Sampling-based methods such as RRT($^*$)~\cite{lavalle2001randomized,karaman2011sampling} and
PRM($^*$)~\cite{kavraki2002probabilistic,karaman2011sampling} handle high-dimensional configuration spaces
but provide only probabilistic guarantees and non-smooth trajectories.  Combinatorial methods
give exact solutions in low dimensions but face
exponential complexity in higher dimensions~\cite{canny1988complexity}. Trajectory optimization
approaches~\cite{ratliff2009chomp, zucker2013chomp, schulman2013finding, betts2010practical} compute
high-quality locally optimal solutions but may become trapped in
local minima.  Our approach inherits the ability of optimization-based
methods to compute certifiably near-optimal smooth trajectories while extending
their reach to problems with complex combinatorial logical structure.

\textbf{Graphs of Convex Sets.}
The GCS framework was introduced by Marcucci et
al.~\cite{marcucci2024shortest}, who developed a tight mixed-integer
convex formulation of the shortest path problem in graphs of convex
sets and showed that the convex relaxation is often exact in practice.
The framework was applied to robot motion planning around obstacles
in~\cite{marcucci2023motion}, establishing GCS as a highly effective
tool for motion planning that produces certifiably near-optimal
trajectories on challenging environments. Subsequent work has extended
the GCS framework towards multi-query settings~\cite{morozov2024multi},
contact-rich manipulation~\cite{graesdal2024towards}, non-Euclidean configuration spaces \cite{cohn2025non}, a unified graph optimization framework \cite{marcucci2025unified}, and other directions \cite{natarajan2024implicit,luna2026augmented, philip2024mixed}.

\textbf{Temporal Logic Motion Planning.}
Using formal specification languages such as linear temporal logic (LTL) to express complex robot tasks
has been widely studied~\cite{baier2008principles, fainekos2009temporal, kress2009temporal, lahijanian2011temporal, plaku2015motion, belta2017formal}. The standard approach constructs a product automaton
between a finite abstraction and discrete finite or Büchi automaton for the LTL
formula, facing doubly exponential complexity.  For continuous
environments, combining formal specifications with continuous
trajectory optimization has been pursued through mixed-integer
programming~\cite{wolff2014optimization}, satisfiability modulo
theories~\cite{shoukry2017linear}, and control barrier
functions~\cite{lindemann2018control, lindemann2020barrier}.  Signal temporal logic
(STL)~\cite{maler2004monitoring, lindemann2018control} extends temporal logic to continuous-time
signals, and has been used with robust and risk-based optimization
formulations~\cite{donze2010robust, safaoui2020control}.

\textbf{Temporal Logic with GCS.}
Combining temporal logic specifications with the GCS framework was
first studied by Kurtz and Lin~\cite{kurtz2023temporal}, who represent
the product of an LTL automaton with a GCS structure as a graph
of convex sets.  This top-down approach handles arbitrary LTL
specifications but inherits doubly exponential automaton construction
cost.  Our approach is complementary: rather than applying a general
pipeline to arbitrary formulas, we identify classes of
specifications whose structure can be directly encoded in the GCS
topology, achieving singly exponential construction in the number of
key-door pairs.  The fragment
library in Section~\ref{sec:variations} explores the
boundary between these two approaches.

\textbf{Task and Motion Planning.} 
Our problem can also be framed within the task and motion planning (TAMP) paradigm~\cite{garrett2021integrated, cambon2009hybrid, kaelbling2011hierarchical, kaelbling2013integrated}, which jointly
reasons about discrete task-level decisions and continuous
motion-level execution.  TAMP methods generally rely on interleaving
symbolic planning with geometric feasibility checking~\cite{cambon2009hybrid,
toussaint2015logic}.  The augmented GCS approach differs from classical
TAMP in that the discrete sequencing problem and the continuous
geometric optimization are solved simultaneously rather than
iteratively, and the solution is certifiably near-optimal rather
than satisficing.

\textbf{Traveling Salesman, Vehicle Routing, and Waypoint Collection.}
The Traveling Salesman Problem (TSP) and its variants formalize the
combinatorial structure underlying waypoint collection
tasks~\cite{applegate2011traveling, laporte1992vehicle}.  The BHK dynamic
programming algorithm~\cite{bellman1962dynamic, held1962dynamic} is the
foundational exact method, running in $O(2^n n^2)$ time.
Several works have studied TSP variants in continuous geometric
spaces, including TSP with neighborhoods \cite{dumitrescu2003approximation}, the Dubins TSP for vehicles with turning
constraints~\cite{ny2011dubins, savla2008traveling} and coverage path
planning~\cite{galceran2013survey}.  Our BHK correspondence result in
Section~\ref{sec:theory} precisely characterizes the relationship
between the augmented GCS and BHK, showing that the augmented GCS
solves a continuous-geometry generalization of both the shortest
Hamiltonian path problem and Traveling Salesman Problem within the same combinatorial complexity
class.

\textbf{Convex Partitioning.}
Computing a convex partition of the free space is a classical problem
in computational geometry~\cite{latombe2012robot, guibas1987linear}.  The cell enumeration
approach based on hyperplane arrangements used in
Algorithm~\ref{alg:partition} follows the classic reverse search method of
Avis and Fukuda~\cite{guibas1987linear}. The optimal complexity
reduction problem for minimizing the size of a decomposition was studied by Geyer et al.~\cite{geyer2008optimal}.

\section{Problem Formulation}
Consider a robot operating in a bounded environment $\mathcal{X} \subset \mathbf{R}^d$. The environment contains a set of obstacles $\mathcal{O} = \{\mathcal{O}_i \}_{i=1,...,n_O}$ with $\mathcal{O}_i \subset \mathcal{X}$. The obstacle-free space is $\mathcal{C} = \mathcal{X}  \setminus \mathcal{O}$. The environment also contains a set of \emph{door} regions $\mathcal{D} = \{\mathcal{D}_i \}_{i=1,...,\nK}$, which are obstacles that can be ``unlocked'' and removed by visiting a corresponding set of \emph{key} regions $\mathcal{K} = \{ \mathcal{K}_i \}_{i=1,...,\nK}$. 
All sets are assumed to be polytopes with given half-space representations:
\begin{equation} \nonumber
\begin{aligned}
\mathcal{X} &= \{ x \in \mathbf{R}^d \mid H_{X} x \leq g_{X} \} \\
\mathcal{O}_i &= \{ x \in \mathbf{R}^d \mid H_{O_i} x \leq g_{O_i} \} \\
\mathcal{D}_i &= \{ x \in \mathbf{R}^d \mid H_{D_i} x \leq g_{D_i} \} \\
\mathcal{K}_i &= \{ x \in \mathbf{R}^d \mid H_{K_i} x \leq g_{K_i} \}.
\end{aligned}
\end{equation}
We assume $\nK$ keys and $\nK$ doors with a bijective pairing: key $i$ unlocks door $i$ for $i=1,...,\nK$. The \emph{key-door mapping} $\mathbf{D}: 2^{\mathcal{K}}\rightarrow 2^\mathcal{D}$ specifies which doors are unlocked by a given key subset. In the base setting, this is the identity mapping $\mathbf D(S) = \{D_i : K_i \in S\}$.



Our goal is to find a trajectory $q$ over a time horizon $T$ that solves the following optimization problem:
\begin{equation} \label{optprob}
    \begin{aligned}
        &\underset{q}{\text{minimize}} \quad &&\alpha \int_0^T \| \dot q(t) \|_2 dt + \beta T \\
        &\text{subject to} \quad && q(t) \in \mathcal{C} \cup \mathcal{K} \cup \mathcal{D} \quad \forall t \in [0, T] \\
        & && q(t) \models \phi(\mathcal{K}, \mathcal{D} ) \\
        & && \dot q(t) \in \mathcal{Q} \quad \forall t \in [0, T] \\
        & && q(0) = q_0, \ q(T) = q_T \\
        & && \dot q(0) = \dot q_0, \ \dot q(T) = \dot q_T.
    \end{aligned}
\end{equation}
The objective is a weighted sum of the trajectory length and the time horizon with weights $\alpha$ and $\beta$, respectively. The first constraint requires the trajectory to be in the union of free space sets, key sets, and door sets for all time. The second constraint is a temporal logic specification that enforces the key-door precedence logic to be satisfied. The third constraint requires the velocity to be in the convex set $\mathcal{Q}$ for all time. The remaining constraints specify initial and terminal conditions for the position and velocity.

The optimization problem \eqref{optprob} is infinite-dimensional. Candidate trajectories are parameterized with piecewise B\'ezier curves defined by a finite set of control points, making the problem finite-dimensional and admitting additional convex constraints or objective terms on higher-order derivatives for dynamic feasibility. This facilitates trajectory design for fully actuated and differentially flat systems.


\subsection{Precedence Constraints via Signal Temporal Logic}
We use signal temporal logic (STL) \cite{maler2004monitoring} to encode the key-door precedence specifications. STL is a variation of temporal logic \cite{baier2008principles} that offers a general and expressive framework to specify complex spatiotemporal tasks and constraints. A STL formula can be composed from the following grammar
\begin{equation}
    \varphi ::= \top \mid p \mid \neg \varphi \mid \varphi_{1} \wedge \varphi_{2} \mid \varphi_{1} \vee \varphi_{2} \mid \varphi_{1} \mathcal{U} \varphi_{2} \mid \varphi_{1} \mathcal{R} \varphi_{2},
\end{equation}
by starting from a set of atomic propositions $AP$ with $p\in AP$ and recursively applying boolean operators: $\neg$ (not), $\wedge$ (and), and $\vee$ (or), and temporal operators: $\mathcal{U}$ (until), and $\mathcal{R}$ (release). The until operator $\phi \mathcal{U} \psi$ is satisfied if $\phi$ remains true until $\psi$ becomes true. The release operator $\phi \mathcal{R} \psi$ is satisfied if $\psi$ remains true until and including when $\phi$ becomes true, and if $\phi$ never becomes true, then $\psi$ always remains true; in other words, $\phi$ releases $\psi$. Additional temporal operators include eventually ($\mathcal{F} \varphi := \top \mathcal{U} \varphi $) and always ($G \varphi := \neg \mathcal{F} \neg \varphi$). 

The base key-door precedence specification uses the release operator:
\begin{equation}
  \label{eq:phi_base}
  \varphi_{\mathrm{base}} = \bigwedge_{i=1}^{\nK}
    \bigl(K_i \, \mathcal{R} \, \lnot D_i\bigr) \land \mathcal{F}(\mathcal{T}),
\end{equation}
which requires $\lnot D_i$ to hold until $K_i$ is visited (and visiting
$K_i$ releases the constraint), while making key collection
\emph{optional} unless geometrically required to reach the goal.
The term $\mathcal{F}(\mathcal{T})$ requires the trajectory to
eventually reach the target set, which is guaranteed by the terminal condition in \eqref{optprob}.

Alternatively, key collection can be made \emph{mandatory} using the until operator:
\begin{equation}
  \label{eq:phi_until}
  \varphi_{\mathcal{U}} = \bigwedge_{i=1}^{\nK}
    \bigl(\lnot D_i \, \mathcal{U} \, K_i\bigr) \land \mathcal{F}(\mathcal{T}),
\end{equation}
which prevents passage through $D_i$ until $K_i$ has been collected
and additionally guarantees $K_i$ is eventually collected.
When the geometry requires all keys to reach the target,
$\varphi_{\mathrm{base}}$ and $\varphi_{\mathcal{U}}$ are equivalent.

\subsection{Graphs of Convex Sets}
The GCS framework \cite{marcucci2024shortest, marcucci2025unified} consists of a graph $\mathcal{G}(\mathcal{V}, \mathcal{E})$ where each vertex $v \in \mathcal{V}$ is associated with a variable $x_v \in \mathbf{R}^{n_v}$, a convex set $\mathcal{X}_v \subseteq \mathbf{R}^{n_v}$, and a convex function $f_v : \mathbf{R}^{n_v} \rightarrow \mathbf{R} $. Each edge $e=(u,v) \in \mathcal{E}$ couples the vertex variables through a convex function $f_e : \mathbf{R}^{n_u + n_v} \rightarrow \mathbf{R}$ and a convex constraint set $\mathcal{X}_e \subseteq \mathbf{R}^{n_u + n_v}$. 

A general GCS optimization problem is
\begin{equation} \label{generalgcs}
    \begin{aligned}
        &\underset{H, \ \{x_v\}_{v\in \mathcal{W}}}{\text{minimize}} \quad &&\sum_{v \in \mathcal{W}} f_v(x_v) + \sum_{e=(u,v) \in \mathcal{F}} f_e(x_u, x_v) \\
        &\text{subject to} \quad && H = (\mathcal{W}, \mathcal{F}) \in \mathcal{H}, \\
        & && x_v \in \mathcal{X}_v,\quad \forall v \in \mathcal{W}, \\
        & && (x_u, x_v) \in \mathcal{X}_e, \quad \forall e = (u, v) \in \mathcal{F}, 
    \end{aligned}
\end{equation}
where the variables are the (discrete) subgraph $H$ with vertex set $\mathcal{W} \subseteq \mathcal{V}$ and edge set $\mathcal{F} \subseteq \mathcal{W}^2 \cap \mathcal{E}$ within an admissible subset of graphs $\mathcal{H}$ and the (continuous) variables $x_v$ for each vertex $v \in \mathcal{W}$.

\textbf{Shortest Path Problem in GCS.}
For a given start vertex $s \in \mathcal{V}$ and target vertex $t \in \mathcal{V}$, a path $p$ is a sequence of distinct vertices that connects $s$ to $t$ via an edge subset $\mathcal F \subset \mathcal E$. The general problem \eqref{generalgcs} can be specialized to the Shortest Path Problem (SPP) by taking $\mathcal{H}$ as the set of all paths, $\mathcal{P}$, from $s$ to $t$ in the graph $\mathcal{G}$. The shortest path problem is:
\begin{equation} \label{spp_gcs}
    \begin{aligned}
        &\underset{p, \ x_v}{\text{minimize}} \quad &&\sum_{e=(u,v) \in \mathcal{E}_p} \ell_e(x_u, x_v) \\
        &\text{subject to} \quad && p \in \mathcal{P}, \\
        & && x_v \in \mathcal{X}_v,\quad \forall v \in p, \\
        & && (x_u, x_v) \in \mathcal{X}_e, \quad \forall e \in \mathcal{E}_p
    \end{aligned}
\end{equation}
Although computationally hard in general, a tight mixed-integer formulation using network flow and perspective cone duality \cite{marcucci2024shortest} enables practically efficient solutions. A convex relaxation obtained by dropping binary constraints, combined with an inexpensive rounding scheme, yields certifiably near-optimal or even optimal solutions for many practical motion planning problems \cite{marcucci2023motion}.


\section{Exact Convex Partitioning and \\Labeled GCS Construction}
This section presents an exact convex partitioning algorithm that represents the traversable space as a union of convex sets and constructs a labeled graph $\mathcal{G} (\mathcal{V},\mathcal{E})$ encoding connectivity among free-space, key, and door regions. Key regions do not constrain traversability: an agent may pass through a key region freely, so key regions are permitted to overlap with free-space cells. Accordingly, only obstacles and doors contribute hyperplanes to the arrangement that partitions the environment. Key regions are incorporated as additional labeled nodes after the partition is constructed, by identifying which cells each key region intersects.

\subsection{Cell Enumeration for Hyperplane Arrangements}
We collect all distinct hyperplanes bounding $\mathcal{X}$ $\mathcal{O}$, and $\mathcal{D}$ into a \emph{hyperplane arrangement} $\mathcal{H} = \{ (h_i, g_i) \}_{i=1}^{n_h}$. Define the sign function $\mathcal{S} : \mathcal{X} \rightarrow \{+, - \}^{n_h}$ by
$$ \mathcal{S}(x) = \begin{cases} - \quad h_i^\top x \leq g_i \\ + \quad h_i^\top x \geq g_i \end{cases}\quad \text{for} \ i=\{1,...,n_h\}.$$
This function assigns to each point in $\mathcal{X}$ a sign pattern, called a \emph{marking}, that encodes which side of each hyperplane that point lies. Each \emph{marking} $m \in \{+,-\}^{n_h}$ defines a (convex) polytopic \emph{cell}
$$ \mathcal{P}_m = \{ x \in \mathcal{X} \mid \mathcal{S}(x) = m \}.$$
These sets form an exact convex partition of the environment $\mathcal{X}$. Cells and their associated markings can be enumerated using the reverse search algorithm \cite{avis1996reverse}, which does a depth-first search to construct a spanning tree of the cells. For $n_h$ hyperplanes in $d$-dimensionals, the cell count is bounded by
$$|M(\mathbf{R}^d)| \leq \sum_{i=0}^{d} {n_h \choose i} = O(n_h^d), $$
with the bound attained when the hyperplanes are in \emph{general position} (no pair of hyperplanes is parallel and no point lies on more than $d$ hyperplanes) \cite{buck1943partition}. This exponential dimension dependence is an inherent computational bottleneck, which is mitigated in structured environments.

\subsection{Labeled Graph Construction}
Associating a vertex with each cell defines the vertex set $\mathcal{V}$. 
Two cells share a facet if and only if their markings differ in exactly one entry, so the edge set is $\mathcal{E} = \{(u,v) \mid d_H(m_u, m_v) = 1\}$, where $d_H$ denotes Hamming distance. 
Door cells can be identified by their markings: for the subset $H \subset \mathcal{H}$ of hyperplanes
defining any door, cells with $m_i = -$ for all $i \in H$
correspond to that door and are merged and labeled accordingly, forming the door vertices $\mathcal{V}_\mathcal{D}$.
Obstacle cells and their associated vertices are similarly identified but then removed. The remaining vertices correspond to free space, forming free-space vertices $\mathcal{V}_\mathcal{C}$ and corresponding polytopes $\mathcal{C}_v$. Free-space nodes can be
further merged to reduce graph size using an optimal complexity
reduction algorithm based on branch and bound, or a greedy
variant~\cite{geyer2008optimal,shaikh2025exact}.

\textbf{Key region labeling.} Unlike obstacles and doors, key regions do not contribute hyperplanes to the arrangement and therefore do not influence the cell structure. After the partition is constructed and free-space cells are merged, each key region $K_i$ is associated with a key vertex $i \in \mathcal{V}_\mathcal{K}$ and connected to the graph by identifying all vertices $v \in \mathcal{V}_\mathcal{C}$ whose associated convex set $\mathcal{C}_v$ satisfies $\mathcal{C}_v \cap K_i \neq \emptyset$. These vertices are connected with an edge to the corresponding key vertex. In the common case where $K_i$ is fully contained within a single merged free-space cell, the key vertex is connected to exactly one free space vertex. When $K_i$ spans multiple cells (for instance, when a key region straddles a cell boundary introduced by a door or obstacle hyperplane), its vertex is connected to multiple free space vertices. The intersection test $\mathcal{C}_v \cap K_i \neq \emptyset$ reduces to a linear feasibility check since both sets are polytopes.

Our algorithm for exact convex partitioning and labeled GCS construction is summarized in Algorithm \ref{partition}. The output is a labeled graph of convex sets $\mathcal{G}(\mathcal{V},\mathcal{E})$ with vertex partition $\mathcal{V} = \{\mathcal{V}_\mathcal{C}, \mathcal{V}_\mathcal{K}, \mathcal{V}_\mathcal{D} \}$. The polytopes corresponding to $\mathcal{V}_\mathcal{C}$ form an exact convex partition of the free space, and the edge set encodes connectivity among free-space, key, and door regions.

    \begin{algorithm}[ht] 
        \caption{Exact Convex Partition \& Labeled GCS Construction}
        \label{alg:partition}
        \begin{algorithmic}[1] \label{partition}
            \renewcommand{\algorithmicrequire}{\textbf{Input:}}
            \renewcommand{\algorithmicensure}{\textbf{Output:}}
            \REQUIRE Environment $\mathcal{X}$, obstacles $\mathcal{O}$, keys $\mathcal{K}$, doors $\mathcal{D}$ \\(all in half-space representation)
            \ENSURE Labeled graph of convex sets $\mathcal{G}( \mathcal{V}, \mathcal{E})$ with $\mathcal{V} = \{\mathcal{V}_\mathcal{C}, \mathcal{V}_\mathcal{K}, \mathcal{V}_\mathcal{D} \}$ and exact convex partition $\mathcal{X} \setminus (\mathcal{O} \cup \mathcal{D}) = \{\mathcal{C}_i \}_{i=1}^{n_C}$
            \STATE Form hyperplane arrangement $\mathcal{H}$ from all boundaries of $\mathcal{X}$, $\mathcal{O}$, $\mathcal{D}$ (key boundaries are excluded)
            \STATE Enumerate non-empty cells and generate markings $M(\mathcal{X})$ (via reverse search \cite{avis1996reverse})
            \STATE Define vertex for each cell to form initial vertex set $\mathcal{V}$
            \STATE Set edge set $\mathcal{E}$ between vertices whose markings differ in exactly one entry
            \FOR {$\mathcal{D}_i \in \mathcal{D}$}
            \STATE Merge cells whose markings place them inside $D_i$; label merged vertex as $D_i$ in $\mathcal{V}_\mathcal{D}$
            \ENDFOR
            \FOR {$\mathcal{O}_i \in \mathcal{O}$}
            \STATE Merge and discard corresponding cells and vertices
            \ENDFOR
            \FOR {$\mathcal{K}_i \in \mathcal{K}$}
            \STATE Create vertex $K_i$ in $\mathcal{V}_\mathcal{K}$. Identify all vertices $v \in \mathcal{V}_\mathcal{C} = \mathcal{V} \setminus  (\mathcal{V}_\mathcal{O} \cup \mathcal{V}_\mathcal{D}) $ whose corresponding convex sets intersect $K_i$ and connect them;
            \ENDFOR
            \STATE Merge free space nodes that preserve convexity (via optimal complexity reduction or greedy algorithm)
            \RETURN $\mathcal{G}( \mathcal{V}, \mathcal{E})$ with $\mathcal{V} = \{\mathcal{V}_\mathcal{C}, \mathcal{V}_\mathcal{K}, \mathcal{V}_\mathcal{D} \}$
        \end{algorithmic}
    \end{algorithm}

\begin{remark}[Partition quality and augmented GCS size]
  \label{rem:partition_quality}
  The number of free-space nodes $|\mathcal{V}_\mathcal{C}|$ produced by
  Algorithm~\ref{alg:partition} directly affects the size of every
  subgraph in the augmented GCS and thus total solve time.
  Investing in better free-space merging (e.g., the branch-and-bound
  method of~\cite{geyer2008optimal}) yields downstream computational
  savings that compound with the number of subgraph
  copies. Excluding key region boundaries from the hyperplane arrangement avoids introducing cuts into the free-space partition that serve no geometric purpose, directly reducing $|\mathcal{V}_\mathcal{C}|$ and thus the size of every subgraph in the augmented GCS.
\end{remark}

\section{Augmented GCS Construction}
In this section, we describe how to construct an augmented graph of convex sets that exactly encodes the key-door precedence constraints based on the labeled GCS $\mathcal{G}(\mathcal{V}, \mathcal{E})$ from Section III. The augmented GCS is built as a layered collection of subgraphs, one per reachable key subset. Each subgraph is a copy of the base graph $\mathcal{G}$, with different edge sets according to the current set of collected keys.  The augmented GCS consists of $\nK + 1$ layers, where each layer represents reachable key subsets of a certain cardinality. Augmented GCS construction begins with a start node connected to a free space vertex whose corresponding set contains the initial condition. By construction, every path connecting the starting node and a target node copy in the augmented GCS satisfies the key-door precedence constraints.  

A simple 2-key environment is shown in Fig. \ref{fig:simple key door}, and the associated augmented GCS is shown in Fig. \ref{fig:augGCS}. We will refer to this as a running example throughout our description.
\begin{figure}[htbp]
  \centering
  \includegraphics[width=0.7\linewidth]{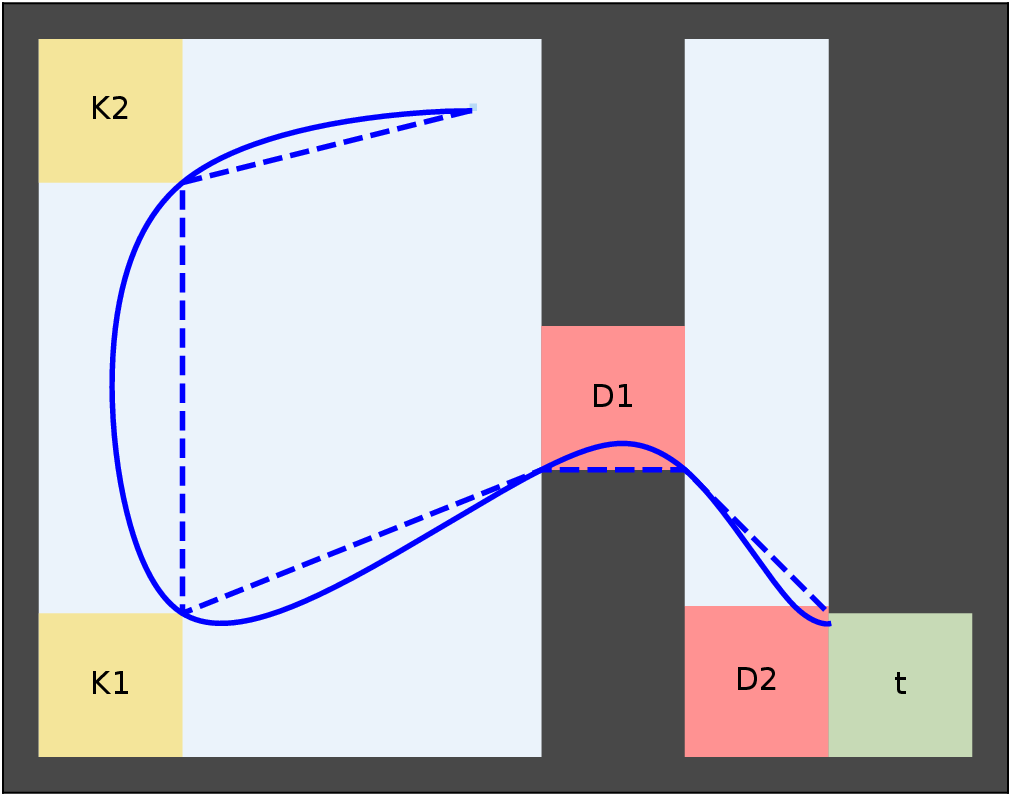}
  \caption{A key-door environment with 2 keys and 2 doors. The dashed blue line shows the optimal linear-piecewise solution, and the solid blue line shows the optimal 4th-order Bézier spline solution with a small penalty on curvature.}
  \label{fig:simple key door}
\end{figure}

\begin{figure}[htbp]
  \centering
  \includegraphics[width=0.9\linewidth]{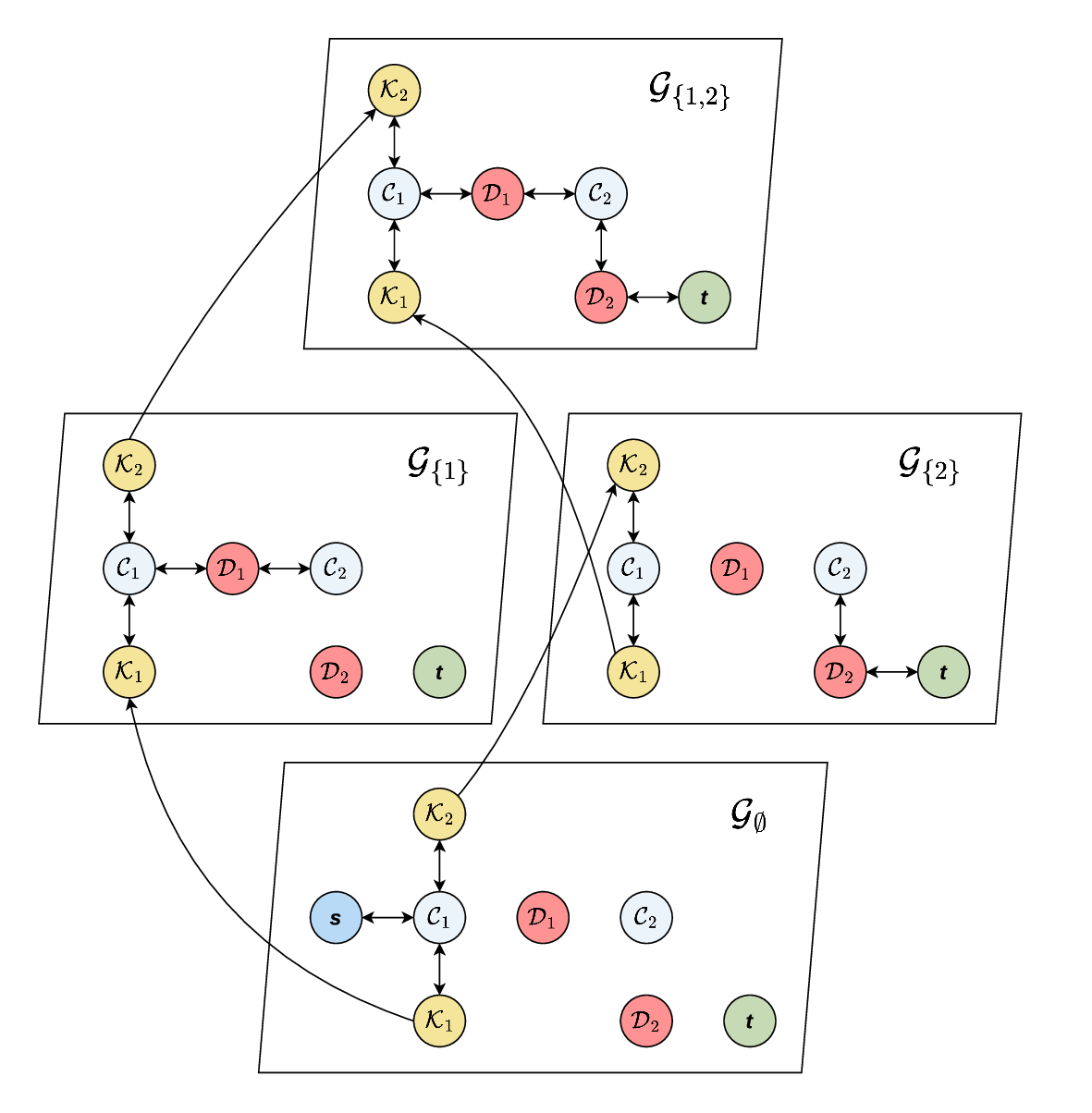}
  \caption{The augmented GCS for a simple 2-key environment. The start point is located in free space node $\mathcal{C}_1$ and the target set is adjacent to $\mathcal{D}_2$.}
  \label{fig:augGCS}
\end{figure}

\subsection{Constructing the Augmented GCS}

\textbf{The Base Layer ($\ell = 0)$.} The base layer of the augmented GCS corresponds to the empty key set $\emptyset \in 2^\mathcal{K}$. The base subgraph $\mathcal{G}_\emptyset(\mathcal{V}_\emptyset, \mathcal{E}_\emptyset)$ has vertex set $\mathcal{V}_\emptyset = \mathcal{V}$ and edge set $$\mathcal{E}_\emptyset = \mathcal{E} \setminus \{ (u, v) \in \mathcal{E} \mid u \in \mathcal{V}_\mathcal{D} \ \text{or} \ v \in \mathcal{V}_\mathcal{D} \},$$
removing all edges incident to door nodes. The base layer for the 2-key environment is shown at the bottom of Fig. \ref{fig:augGCS}.

\textbf{Layer $\ell$ ($1 \leq \ell < n_{\mathcal{K}}$).} Layer $\ell$ consists of subgraphs $\mathcal{G}_S$ for each reachable $\ell$-element key subset $S \in S_\ell \subset 2^\mathcal{K}$. The reachable key subsets are found with breadth- or depth-first search from each subgraph at the previous layer. The \emph{width} of layer $\ell$ is $|S_\ell|$. Each subgraph $\mathcal{G}_S(\mathcal{V}_S, \mathcal{E}_S)$ has vertex set $\mathcal{V}_S =\mathcal{V}$ and edge set $$\mathcal{E}_S = \mathcal{E}_\emptyset \cup \{ (u, v) \in \mathcal{E} \mid u \in \mathbf{D}(S) \ \text{or} \ v \in \mathbf{D}(S) \},$$
reinserting edges incident to doors unlocked by keys in $S$. For each key node $K_v \in S$, a directed edge of zero cost is added from the copy of key $K_v$ in $\mathcal{G}_{S - \{v\} }$ to its copy in $\mathcal{G}_{S}$, encoding the key collection event.

In our running 2-key example, both keys are reachable from the start node, so $S_1 = \{ \{1\}, \{2\} \}$ and layer 1 has width 2. The subgraphs $G_{\{1\}}$ and $G_{\{2\}}$ are shown in the middle of Fig. \ref{fig:augGCS}, along with the directed edges from key nodes in $G_\emptyset$.

\textbf{Layer $\nK$.} The final layer corresponds to the full key set $\mathcal{K} \in 2^\mathcal{K}$ and consists of a single subgraph $\mathcal{G}_\mathcal{K}(\mathcal{V}_\mathcal{K}, \mathcal{E}_\mathcal{K})$ with vertex set $\mathcal{V}_\mathcal{K}=\mathcal{V}$ and edge set $\mathcal{E}_\mathcal{K} = \mathcal{E}$, the same edge set as $\mathcal{G}$. For each key node $v \in \mathcal{V}_\mathcal{K}$, a directed edge of zero cost is added from the key node copy in $\mathcal{G}_{\mathcal{K} - \{v\} }$.

In our 2-key example,  the subgraph $G_{\{1, 2\}}$ associated with the full key set is shown at the top of Fig. \ref{fig:augGCS}.

\textbf{Target Merging.} Finally, given a target node $t \in \mathcal{V}_{\mathcal{C}}$ whose corresponding set contains the terminal condition, we merge all its copies node throughout all layers and subgraphs into a single node that serves as the target node in the augmented GCS. The algorithm for constructing the augmented graph is summarized in Algorithm \ref{Algo1}.

For each subgraph in the augmented graph, we associate the free space, key, and door regions with the corresponding node copies to obtain the augmented GCS. The machinery of the GCS motion planning framework \cite{marcucci2023motion,marcucci2024shortest} can be applied to compute a shortest path from the start node to the (merged) target node in the augmented GCS. 

    \begin{algorithm}[ht] 
        \caption{Augmented GCS Construction}
        \begin{algorithmic}[1] \label{Algo1}
            \renewcommand{\algorithmicrequire}{\textbf{Input:}}
            \renewcommand{\algorithmicensure}{\textbf{Output:}}
            \REQUIRE Labeled graph $\mathcal{G}(\mathcal{V}, \mathcal{E})$, with $\mathcal{V} = \{\mathcal{V}_\mathcal{C}, \mathcal{V}_\mathcal{D}, \mathcal{V}_\mathcal{K} \}$, start node $s$ connected to $v \in \mathcal{V}_\mathcal{C}$, target node $t \in \mathcal{V}_\mathcal{C}$
            \ENSURE Augmented graph $\hat{\mathcal{G}}( \hat{\mathcal{V}}, \hat{\mathcal{E}})$
            \STATE Create subgraph $\mathcal{G}_\emptyset(\mathcal{V}_\emptyset, \mathcal{E}_\emptyset)$, where $\mathcal{V}_\emptyset := \mathcal{V}$, $\mathcal{E}_\emptyset := \mathcal{E} - \{ (u, v) \in \mathcal{E} \mid u \in \mathcal{V}_\mathcal{D} \ \text{or} \ v \in \mathcal{V}_\mathcal{D} \}$ 
            \FOR {$\ell=1$ to $\nK$}
            \STATE Find all reachable $\ell$-element key sets $S_\ell \subset 2^\mathcal{K}$ from start node copy in all subgraphs with $(\ell-1)$-element key sets
            \FOR {$S \in S_\ell$}
            \STATE Create subgraph $\mathcal{G}_S(\mathcal{V}_S, \mathcal{E}_S)$, where $\mathcal{V}_S := \mathcal{V}$, $\mathcal{E}_S := \mathcal{E}_\emptyset \cup \{ (u, v) \in \mathcal{E} \mid u \in \mathbf{D}(S) \ \text{or} \ v \in \mathbf{D}(S) \} $
            \FOR {$v \in S$}
            \STATE Add directed edge $\mathcal{E}_S \leftarrow \mathcal{E}_S \cup \{ (u, v) \} $, where $u$ is the copy of $v$ in  $\mathcal{V}_{S - \{ v \} }$
            \ENDFOR
            \ENDFOR
            \ENDFOR
            \STATE Merge copies of target node in all layers \& subgraphs
            \RETURN $\hat{\mathcal{G}}( \hat{\mathcal{V}}, \hat{\mathcal{E}})$, where $ \hat{\mathcal{V}} = \cup_{S \in \{ S_\ell \}_{\ell = 0}^{\nK}} \mathcal{V}_S $, $ \hat{\mathcal{E}} = \cup_{S \in \{ S_\ell \}_{\ell = 0}^{\nK}} \mathcal{E}_S $ 
        \end{algorithmic}
    \end{algorithm}

\textbf{Size of the augmented GCS.}
The total number of subgraphs in the augmented GCS satisfies
\begin{equation}
\sum_{\ell=0}^{\nK} |S_\ell| \leq \sum_{\ell=0}^{\nK} {\nK \choose \ell} = 2^{\nK}.
\end{equation}
with equality when all keys are reachable from the start node $s$. This exponential dependence on the number of keys represents an important computational challenge. However, the subgraph lattice structure has a close correspondence with the landmark Bellman-Held-Karp algorithm, which will be made precise in Section VI. Fig. \ref{fig:subgraphs} shows subgraph lattices for up to 4 keys when all key subsets are reachable. 

At another extreme, when the environment geometry requires keys to be collected in a fixed sequence, the augmented GCS reduces to a chain of $\nK + 1$ subgraphs. There are degenerate cases where some keys are not reachable at all, which reduces the count further. In intermediate problems, the geometry constrains the reachable key subsets and limits the size of the augmented GCS. The \emph{maximum width} $\max_\ell |S_\ell|$ of the augmented graph is the key computational bottleneck. 

\begin{figure}[htbp] \label{keysubgraphs}
  \centering
  \includegraphics[width=\linewidth]{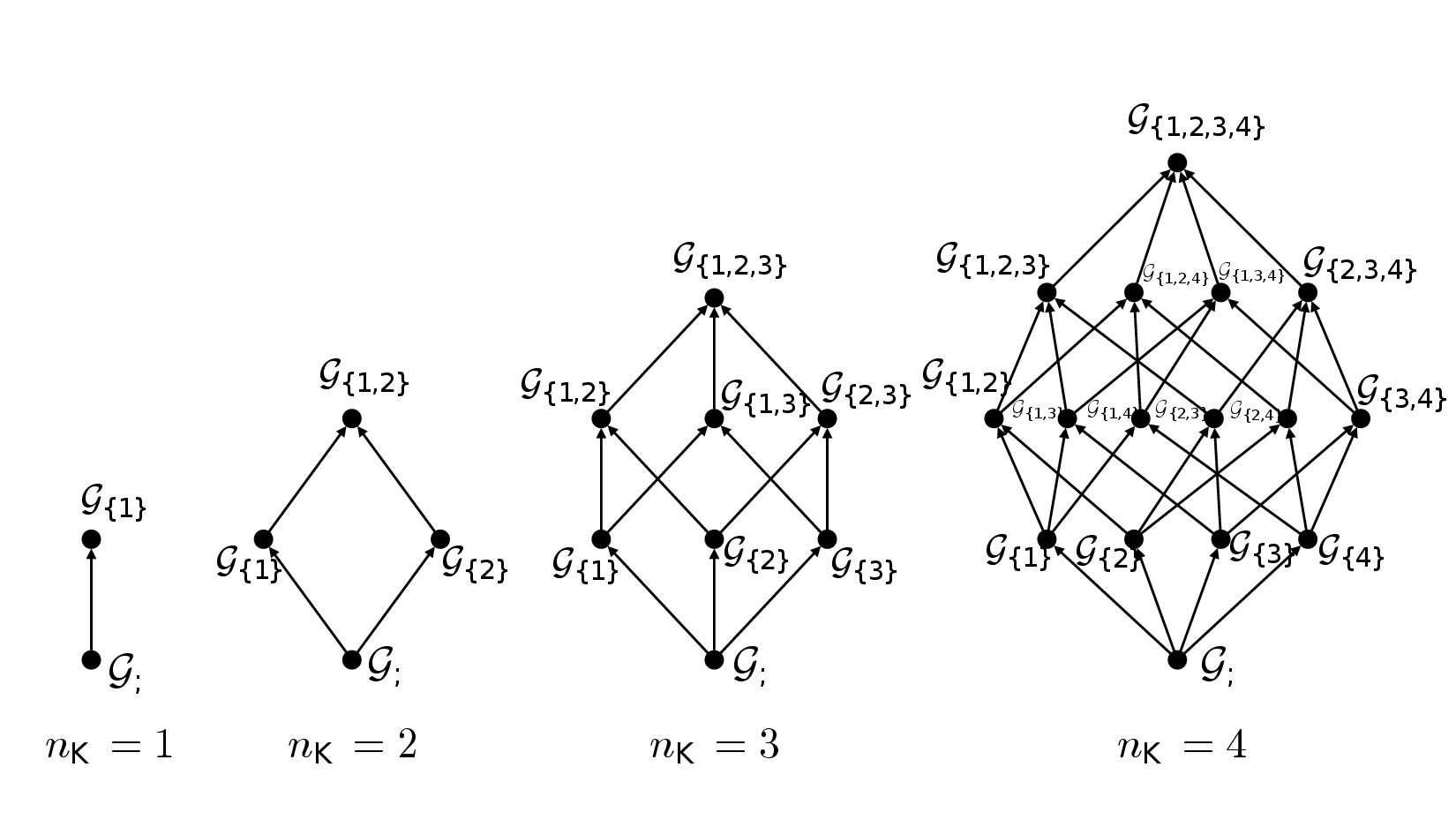}
  \caption{Augmented GCS layer structure for base graphs with up to $\nK = 4$ keys.}
  \label{fig:subgraphs}
\end{figure}

\section{Optimality of an Augmented GCS Shortest Path}
In this section, we formally prove that a shortest path in the augmented GCS corresponds to an optimal solution of the original problem \eqref{optprob}. Thus, the mixed-integer convex reformulation of \cite{marcucci2024shortest,marcucci2023motion} applied to the augmented GCS exactly solves \eqref{optprob}, up to the finite B\'ezier parameterization of the trajectory $q$. We will study the tightness of the convex relaxation in our numerical experiments. In the absence of a bound on the horizon $T$, an optimal solution is found if one exists. Otherwise, infeasibility is certified when the merged target node is disconnected from the start node $s$ in the augmented graph.

We now state the main optimality result. 

\begin{theorem}[Optimality of Augmented GCS Shortest Path]\label{thm:main}
A shortest path from start node $s$ to the merged target node $t$ in $\hat{\mathcal{G}}(\hat{\mathcal V}, \hat{\mathcal E})$ corresponds to an optimal solution of problem~\eqref{optprob}, up to a finite parameterization of the trajectory $q$ by B\'{e}zier curves.
\end{theorem}

The proof proceeds via three lemmas establishing soundness, completeness, and cost equivalence. The proofs for these lemmas are given in Appendix~\ref{app:proofs}.


\begin{lemma}[Soundness]\label{lem:sound}
Every path in $\hat{\mathcal{G}}$ from start node $s$ to merged target node $t$ corresponds to a trajectory that satisfies the key-door precedence specification $\phi$.
\end{lemma}

\begin{lemma}[Completeness]\label{lem:complete}
Every feasible trajectory $q^*$ of problem~\eqref{optprob} corresponds to some path in $\hat{\mathcal{G}}$.
\end{lemma}

\begin{lemma}[Cost Equivalence]\label{lem:cost}
The cost of a path $\hat{p}$ in $\hat{\mathcal{G}}$ equals the cost of the corresponding trajectory in problem~\eqref{optprob}.
\end{lemma}

\begin{proof}[Proof of Theorem~\ref{thm:main}]
By Lemmas~\ref{lem:sound} and~\ref{lem:complete} together establish a bijection between paths in $\hat{\mathcal{G}}$ and feasible trajectories of \eqref{optprob}. By Lemma~\ref{lem:cost}, both objectives coincide over this bijection. Therefore:
\[
\min_{\hat{p} \in \hat{\mathcal{P}}}\  \text{cost}(\hat{p}) = \min_{q \text{ feasible for } \eqref{optprob}} \left( \alpha \int_0^T \|\dot{q}\|^2\,dt + \beta T \right),
\]
where $\hat{\mathcal{P}}$ denotes the set of all paths from $s$ to $t$ in $\hat{\mathcal{G}}$. A shortest path in $\hat{\mathcal{G}}$ thus recovers an optimal solution to \eqref{optprob}, up to the finite B\'{e}zier curve parameterization of $q$.
\end{proof}

\begin{remark}[Trajectory Continuity at Layer Transitions]
When $\hat{p}$ crosses a directed inter-layer edge at a key node, the corresponding trajectory transitions between two copies of the same convex key set. Since both copies share identical geometry and the transition incurs zero positional displacement, the trajectory $q$ remains $C^k$-continuous for any finite order $k$ supported by the B\'{e}zier parameterization. This continuity condition must be explicitly included in the B\'{e}zier curve boundary constraints.

\end{remark}

\begin{remark}[Repeated Key Visits]
If a trajectory visits the same key region $K_i$ multiple times, subsequent visits do not alter the accumulated key set $S_j$. The path $\hat{p}$ remains in the current subgraph without triggering additional layer transitions, which is consistent with the construction. A unique subgraph assignment is obtained by using the first visit to each key as the layer-transition trigger.
\end{remark}
 
\begin{remark}[Finite Parameterization]
It is worth emphasizing that optimality holds up to the B\'{e}zier curve parameterization of $q$. Strictly, Theorem~\ref{thm:main} recovers a finitely parameterized optimal trajectory, not a global optimum over all smooth curves. The gap introduced by finite-order B\'{e}zier parameterization is a separate approximation not addressed by the augmented GCS construction. However, the gap is small in practice for moderate B\'{e}zier order.
\end{remark}

\section{Correspondence with the Bellman--Held--Karp Algorithm}
\label{sec:theory}

The Bellman--Held--Karp (BHK) algorithm~\cite{bellman1962dynamic,held1962dynamic} solves the Traveling Salesman Problem (TSP), finding a minimum cost tour amongst $n$ cities in $O(2^nn^2)$ time via dynamic programming. Let $g(S, v)$ denote the minimum cost of a path that starts at a fixed origin $0$, visits exactly the cities in $S \subseteq \{1, \ldots, n\}$, and ends at city $v \in S$. The recurrence is:
\begin{align}
g(\{v\}, v) &= c_{0v}, \label{eq:bhk_base}\\
g(S, v) &= \min_{u \in S \setminus \{v\}} \left[ g(S \setminus \{v\}, u) + c_{uv} \right], \label{eq:bhk_rec}
\end{align}
and the optimal tour cost is $\min_{v \neq 0}\left[g(\{1,\ldots,n\}, v) + c_{v0}\right]$.

To demonstrate the cleanest correspondence between BHK and the augmented GCS structure, we consider the until-based specification for mandatory key collection in \eqref{eq:phi_until}, and further assume there are no doors, so that the door constraints are trivially satisfied. The specification then becomes 
\begin{equation}
  \label{eq:TSPspec}
  \varphi_{\mathcal{U}} = \bigwedge_{i=1}^{\nK}
    \bigl(\top \, \mathcal{U} \, K_i\bigr) \land \mathcal{F}(\mathcal{T}) = \bigwedge_{i=1}^{\nK}
    \bigl( \mathcal{F}(K_i) \bigr) \land \mathcal{F}(\mathcal{T})
\end{equation}
requiring all keys to eventually be collected but not specifying a collection order.

\subsection{Structural Correspondence and Correspondence Theorem}

\begin{definition}[BHK--GCS Dictionary]
Table~\ref{tab:bhk_dict} summarizes the correspondence between BHK and the augmented GCS:

\begin{center}
\begin{table}[h]
  \caption{BHK--GCS Correspondence Dictionary}
  \label{tab:bhk_dict}
  \centering
  \begin{tabular}{ll}
    \toprule
    \textbf{BHK (TSP)} & \textbf{Augmented GCS} \\
    \midrule
    City $i$ & Key node $K_i \in \Vcal_\Kcal$ \\
    City subset $S$ & Key subset $S \in 2^\Kcal$ \\
    Subproblem state $(S,v)$ & Subgraph $\Gcal_S$, active vertex $v$ \\
    Transition cost $c_{uv}$ & GCS edge cost $\ell_e(x_u, x_v)$ \\
    Discrete pairwise distances & Shortest path through continuous geometry \\
    DP table entry $g(S,v)$ & Shortest path cost to $K_v^{(S)}$ in $\Ghat$ \\
    \bottomrule
  \end{tabular}
\end{table}
\end{center}
\end{definition}

\begin{theorem}[BHK--Augmented GCS Correspondence]\label{thm:bhk}
Let the environment contain $\nK$ keys and no doors, with all keys reachable from the start node. Then:
\begin{enumerate}
    \item The layer structure of $\hat{\mathcal{G}}$ implements exactly the BHK subset enumeration over $2^{\nK}$ key subsets.
    \item The inter-layer directed edges of $\hat{\mathcal{G}}$ implement exactly the BHK recurrence relation~\eqref{eq:bhk_rec}.
    \item A shortest path in $\hat{\mathcal{G}}$ solves the same optimal subset-sequencing problem as BHK, generalized from discrete transition costs to shortest paths through continuous geometry.
\end{enumerate}
\end{theorem}


\begin{proof}[Proof of 1: Layer Structure = BHK Subset Enumeration]
In BHK, the dynamic programming table is indexed by all subsets $S \subseteq \{1,\ldots,n\}$, processed in order of increasing cardinality $|S| = 0, 1, 2, \ldots, n$.

In the augmented GCS, layer $\ell$ consists of subgraphs $\{\mathcal{G}_S : S \in \mathcal{S}_\ell\}$, one per reachable $\ell$-element key subset. The total number of subgraphs satisfies:
\[
\sum_{\ell=0}^{\nK} |\mathcal{S}_\ell| \leq \sum_{\ell=0}^{\nK} \binom{\nK}{\ell} = 2^{\nK},
\]
with equality when all keys are reachable from the start node. The subgraphs at layer $\ell$ are in bijection with the BHK subproblems of cardinality $\ell$, and the construction processes them in order of increasing $\ell$. The enumeration structures are identical.
\end{proof}

\begin{proof}[Proof of 2: Directed Inter-Layer Edges = BHK Recurrence]
The BHK recurrence~\eqref{eq:bhk_rec} encodes the optimal cost to reach state $(S, v)$ is achieved by choosing the best predecessor $u$, taking the optimal path to state $(S \setminus \{v\}, u)$, then paying transition cost $c_{uv}$.

In Algorithm~2 (lines 6--8), for each $S \in \mathcal{S}_\ell$ and each $K_v \in S$, a directed edge is added from $K_v^{(S \setminus \{v\})}$ in $\mathcal{G}_{S \setminus \{v\}}$ to $K_v^{(S)}$ in $\mathcal{G}_S$. The augmented GCS shortest path cost to $K_v^{(S)}$ satisfies:
\begin{multline}\label{eq:gcs_rec}
\hat{g}(S, K_v) = \min_{u \in S \setminus \{v\}} \Big[
\hat{g}(S \setminus \{v\}, K_u) \\
+ d_{\hat{\mathcal{G}}}(K_u^{(S \setminus \{v\})}, K_v^{(S \setminus \{v\})})
\Big].
\end{multline}
where $\hat{g}(S, K_v)$ is the shortest path cost to $K_v^{(S)}$ in $\hat{\mathcal{G}}$ and $d_{\hat{\mathcal{G}}}(\cdot, \cdot)$ denotes intra-subgraph shortest path distance through the convex free space. Equation~\eqref{eq:gcs_rec} is the BHK recurrence~\eqref{eq:bhk_rec} with discrete transition costs $c_{uv}$ replaced by continuous geodesic distances $d_{\hat{\mathcal{G}}}(K_u^{(S\setminus\{v\})}, K_v^{(S\setminus\{v\})})$ through the obstacle-laden convex geometry.
\end{proof}

\begin{proof}[Proof of Part 3: Optimal Path = Generalized BHK Solution]
Consider the restriction of the augmented GCS shortest path problem to key-node-to-key-node transitions. Collapse each subgraph $\mathcal{G}_S$ to a complete graph on its key nodes, with edge weights given by intra-subgraph shortest path distances. The resulting combinatorial problem is:
\[
\min_{\text{sequence } (K_{i_1}, \ldots, K_{i_m})} \sum_{j=0}^{m} d_{\mathcal{G}}(v_j, K_{i_{j+1}}),
\]
where $v_0 = s$, $v_{m+1} = t$, and $d_\mathcal{G}$ denotes free-space shortest path distance. This is a Shortest Hamiltonian Path problem on a subset of key sets with geodesic distances through convex obstacle-laden geometry replacing Euclidean city-to-city distances. A pure TSP tour can be obtained by setting the boundary constraints in \eqref{optprob} with $q(0) = q(T)=q_0\in\Kcal_1$.

The augmented GCS shortest path problem therefore subsumes BHK: it solves the combinatorial key-sequencing problem (as BHK does for city ordering) and simultaneously optimizes the continuous trajectory within each sequencing decision. The two problems coincide precisely when the intra-subgraph geodesics are replaced by precomputed pairwise distances for singleton key regions, recovering the standard BHK setting.
\end{proof}

\subsection{Key Generalizations Beyond BHK}

The augmented GCS framework strictly generalizes BHK along three axes.

\paragraph{Specification}
Classical TSP requires visiting all cities, which is analogous to the until-based specification with no doors. The augmented GCS framework allows much richer specifications via alternative temporal logic fragments, several of which will be discussed further in the next section. Including door regions incorporates precedence structure, and the release operator $K_i \,\mathcal{R}\, \neg D_i$ makes key collection optional unless geometrically necessary to reach the goal. This corresponds to a \emph{generalized TSP on a subset}: the algorithm self-selects which keys to collect based on whether doing so shortens the path to the terminal condition. 

\paragraph{Geometry}
    BHK operates on a complete graph with precomputed pairwise distances. The augmented GCS operates directly on the continuous free space, computing geodesic distances through convex obstacle-free regions as part of the same optimization. Crucially, when door regions are present, the effective ``distance'' between two keys depends on which doors have been unlocked — a state-dependent geometric coupling that BHK, which takes fixed distances as input, cannot express.

\paragraph{Continuous-Combinatorial Coupling}
In BHK, the combinatorial sequencing problem and the metric distances are decoupled; distances are fixed inputs. In the augmented GCS, the two are solved jointly: the convex optimization over B\'{e}zier curve control points within each subgraph is coupled to the shortest path search across subgraphs through shared vertex variables at subgraph boundaries. This coupling enables certifiably near-optimal solutions via convex relaxation, a guarantee that BHK does not provide in non-metric settings.

\subsection{Complexity Alignment}
The augmented GCS operates within the same combinatorial complexity class as BHK. The worst-case number of subgraphs in the augmented GCS is $2^{\nK}$, matching the $O(2^n)$ BHK table size. The additional $O(n^2)$ factor in BHK's $O(2^n n^2)$ complexity corresponds to the $\nK$ inter-layer edge additions per subgraph and the intra-subgraph shortest path solves, both of which are polynomial in $\nK$ for fixed graph structure. By incorporating continuous geometry, the shortest path problem on the augmented GCS, via exact mixed-integer convex reformulation, inherits the same additional computational complexity challenges (NP-hardness) as the base GCS framework \cite{marcucci2024shortest}. However, in instances where the optimality gap is small between convex relaxation and rounding schemes, the continuous geometry optimization is carried out at no additional exponential cost.

\section{Specification Variations}
\label{sec:variations}

We develop eight variations of the base key-door precedence
specification, each obtained by modifying the STL formula and the
corresponding augmented GCS construction.  For each variation we
state the specification, describe the modification to
Algorithm~\ref{Algo1}, summarize the correctness result, and
note the worst-case subgraph count.  Full proofs of all correctness theorems appear in Appendix ~\ref{app:variation_proofs}.

\subsection{Variation 1: Mandatory Key Collection (Until Operator)}
\label{subsec:until}

\subsubsection*{Specification}
Replace the release operator with the until operator:
\begin{equation}
  \label{eq:phi_U}
  \varphi_{\mathcal{U}} = \bigwedge_{i=1}^{\nK}
    (\lnot D_i \, \mathcal{U} \, K_i) \land \mathcal{F}(\mathcal T).
\end{equation}
This is strictly stronger than~\eqref{eq:phi_base} when keys are not
geometrically required: it both prevents passage through $D_i$ before
$K_i$ is collected and guarantees $K_i$ is eventually collected.

\subsubsection*{Augmented graph modification}

\begin{construction}
  \label{con:until}
  The layer structure and inter-layer directed edges are identical to
  Algorithm~\ref{Algo1}.  The sole modification is to
  line~11: instead of merging all copies of $t$ across all layers,
  instantiate the target node only in the top-layer subgraph
  $\Gcal_\Kcal$:
  \[
    t^* = t^{(\Kcal)} \in \Gcal_\Kcal.
  \]
  All other copies $t^{(S)}$ for $S \subsetneq \Kcal$ are removed
  from $\Ghat$.
\end{construction}

\begin{theorem}
  \label{thm:until}
  A shortest path from start node $s$ to target node $t^*$ in the augmented GCS of
  Construction~\ref{con:until} corresponds to an optimal solution
  of~\eqref{optprob} with specification $\varphi_{\mathcal{U}}$.
\end{theorem}
\begin{proof} See Appendix~\ref{app:variation_proofs}. \end{proof}

\textit{Worst-case subgraph count:} $O(2^{\nK})$, unchanged from the base construction.

\begin{remark}
  When the geometry forces all keys to be collected,
  $\varphi_{\mathrm{base}}$ and $\varphi_{\mathcal{U}}$ are
  equivalent and both augmented graphs yield the same optimal
  solution.
\end{remark}

\subsection{Variation 2: Pure Wayset Collection (TSP Variation)}
\label{subsec:tsp}

As described in the previous section, this variation removes all door regions and requires the agent to
visit every key region (wayset) before reaching the goal.  It is the
cleanest instantiation of the BHK correspondence: the augmented GCS
implements BHK exactly, without the state-dependent geometry effect
of door unlocking.

\subsubsection*{Specification}
\begin{equation}
  \label{eq:phi_tsp}
  \varphi_{\mathrm{TSP}} = \bigwedge_{i=1}^{\nK}
    \mathcal{F}(K_i) \land \mathcal{F}(\mathcal T).
\end{equation}
Equivalently, $\varphi_{\mathrm{TSP}} = \bigwedge_i (\top\,
\mathcal{U}\, K_i) \land \mathcal{F}(\mathcal T)$.  No door regions are present in the environment.

\subsubsection*{Augmented graph modification}

\begin{construction}
  \label{con:tsp}
  Remove all door nodes and edges from $\Gcal$; the remaining graph
  has $\mathcal V = \{\mathcal V_{\mathcal C}, \mathcal V_{\mathcal K}\}$.  Apply Construction~\ref{con:until}
  (target restricted to the top layer), with the door-reinsertion
  step (Algorithm~\ref{Algo1}, line~5) vacuous since no doors
  exist.  The result is $\Gcal_\Kcal$ at the top layer being
  identical to the original graph, and all inter-layer edges being
  zero-cost transitions at key nodes.
\end{construction}

\begin{theorem}
  \label{thm:tsp}
  A shortest path from $s$ to $t^*$ in Construction~\ref{con:tsp}
  corresponds to an optimal solution of the continuous-geometry
  shortest Hamiltonian path problem: visiting all $\nK$ waysets
  in obstacle-free convex space with minimum total cost.
\end{theorem}
\begin{proof} See Appendix~\ref{app:variation_proofs}. \end{proof}

\textit{Worst-case subgraph count:} $O(2^{\nK})$.

\begin{remark}
The Hamiltonian path begins at $q_0$, visits each wayset, and ends at $q_T$. To produce a pure TSP tour among waysets, the boundary constraints in \eqref{optprob} can be replaced with $q(0) = q(T)=q_0\in\Kcal_1$. In this case, where a wayset is the source, the number of subgraphs reduces from $2^{\nK}$ to $2^{\nK-1}$.
\end{remark}

\begin{remark}[$k$-of-$n$ wayset collection and connections to
  combinatorial optimization]
  \label{rem:k_of_n}
 A natural generalization of Variation 2 requires the robot to visit \emph{any} $k$ out of the $\nK$ key regions before reaching the target. The augmented graph modification is simply to truncate the subgraph lattice after layer $k$ and merge target nodes across all subgraphs at that layer. The shortest path selects the optimal $k$-element key subset to collect. The worst-case subgraph count reduces from
  $O(2^{n_\mathcal{K}})$ to $O(\sum_{\ell=0}^{k} \binom{n_\mathcal{K}}{\ell})
  = O(n_\mathcal{K}^k)$ for fixed $k$, a substantial reduction when
  $k \ll n_\mathcal{K}$.
  
  In combinatorial optimization, the discrete
  counterpart of this variation is the \emph{$k$-TSP} (selective
  TSP), which seeks the minimum-cost Hamiltonian
  path visiting exactly $k$ out of $n$ cities.  The augmented GCS provides
  a continuous-geometry generalization of $k$-TSP.  
\end{remark}

\subsection{Variation 3: Ordered Key Collection}
\label{subsec:ordered}

\subsubsection*{Specification}
Keys must be collected in the fixed sequence $K_1, K_2, \ldots,
K_{\nK}$:
\begin{equation}
  \label{eq:phi_ord}
  \varphi_{\mathrm{ord}} = (\lnot D_1 \, \mathcal{U} \, K_1)
  \land \bigwedge_{i=2}^{\nK}
  \bigl(\lnot D_i \, \mathcal{U} \, (K_{i-1} \land \mathcal{F} K_i)\bigr)
  \land \mathcal{F}(\mathcal T).
\end{equation}

\subsubsection*{Augmented graph modification}

\begin{construction}
  \label{con:ordered}
  The augmented graph is a \emph{linear chain} of $\nK + 1$ subgraphs:
  $\Gcal_\emptyset \to \Gcal_{\{1\}} \to \Gcal_{\{1,2\}} \to \cdots
  \to \Gcal_{\{1,\ldots,\nK\}}$.  At layer $\ell$, create only the
  single subgraph $\Gcal_{\{1,\ldots,\ell\}}$ and add a directed edge
  only at $K_\ell$.  The target $t^*$ is instantiated only in
  $\Gcal_{\{1,\ldots,\nK\}}$.
\end{construction}

\begin{theorem}
  \label{thm:ordered}
  A shortest path from $s$ to $t^*$ in Construction~\ref{con:ordered}
  corresponds to an optimal solution of~\eqref{optprob} with
  $\varphi_{\mathrm{ord}}$.
\end{theorem}
\begin{proof} See Appendix~\ref{app:variation_proofs}. \end{proof}

\begin{remark}
  The ordered construction requires exactly $\nK + 1$ subgraphs,
  reducing the worst-case augmented GCS size from $O(2^{\nK})$ to
  $O(\nK)$ subgraphs. When key order is only partially specified, the augmented graph is
  a partial-order directed acyclic graph (DAG) on key subsets with complexity intermediate
  between $O(\nK)$ and $O(2^{\nK})$.
\end{remark}

\subsection{Variation 4: Disjunctive Keys}
\label{subsec:disjunctive}

\subsubsection*{Specification}
Either key $K_i^A$ or key $K_i^B$ suffices to unlock door $D_i$:
\begin{equation}
  \label{eq:phi_disj}
  \varphi_{\mathrm{disj}} = \bigwedge_{i=1}^{n_D}
    \bigl((K_i^A \lor K_i^B) \, \mathcal{R} \, \lnot D_i\bigr)
    \land \mathcal{F}(\mathcal T).
\end{equation}

\subsubsection*{Augmented graph modification}

\begin{construction}
  \label{con:disj}
  Introduce abstract keys $\tilde{\Kcal} = \{\tilde{K}_1, \ldots,
  \tilde{K}_{n_D}\}$.  Assign both $K_i^A$ and $K_i^B$ the abstract
  label $\tilde{K}_i$.  Build the subset lattice over $\tilde{\Kcal}$
  rather than individual keys.  For each $\tilde{K}_i \in S$, add
  directed inter-layer edges at both physical key nodes:
  $K_i^{A,(S\setminus\{\tilde{K}_i\})} \to K_i^{A,(S)}$ and
  $K_i^{B,(S\setminus\{\tilde{K}_i\})} \to K_i^{B,(S)}$.
\end{construction}

\begin{theorem}
  \label{thm:disj}
  A shortest path from $s$ to $t$ in Construction~\ref{con:disj}
  corresponds to an optimal solution of~\eqref{optprob} with
  $\varphi_{\mathrm{disj}}$.
\end{theorem}
\begin{proof} See Appendix~\ref{app:variation_proofs}. \end{proof}

\textit{Worst-case subgraph count:} $O(2^{n_D})$ over abstract keys.  The only
overhead relative to the base construction is a doubling of
inter-layer edges per abstract key.

\begin{remark}
  This construction easily generalizes to more than two disjunctive keys. This variation also enables the framework to incorporate non-convex key regions that can be decomposed into a union of polytopes; the components then correspond to a set of disjunctive keys.
\end{remark}

\subsection{Variation 5: Multi-Key Doors (Conjunctive Keys)}
\label{subsec:multikey}

\subsubsection*{Specification}
Door $D_i$ requires a specified subset $\mathcal{K}_i \subseteq
\Kcal$ of keys:
\begin{equation}
  \label{eq:phi_multi}
  \varphi_{\mathrm{multi}} = \bigwedge_{i=1}^{n_D}
    \Bigl(\bigwedge_{j:\, K_j \in \mathcal{K}_i} K_j\Bigr)
    \mathcal{R}\, \lnot D_i \land \mathcal{F}(\mathcal T).
\end{equation}

\subsubsection*{Augmented graph modification}
The sole modification to Algorithm~\ref{Algo1} is the
door-unlocking mapping:
\begin{equation}
  \label{eq:multi_unlock}
  \mathbf D(S) = \{D_i : \mathcal{K}_i \subseteq S\}.
\end{equation}
Edges incident to $D_i$ are reinserted only when $S$ contains the
\emph{entire} required key set $\mathcal{K}_i$.

\begin{theorem}
  \label{thm:multi}
  A shortest path from $s$ to $t$ in the augmented GCS with
  mapping~\eqref{eq:multi_unlock} corresponds to an optimal solution
  of~\eqref{optprob} with $\varphi_{\mathrm{multi}}$.
\end{theorem}
\begin{proof} See Appendix~\ref{app:variation_proofs}. \end{proof}

\textit{Worst-case subgraph count:} $O(2^{\nK})$, unchanged. This variation
composes naturally with Variation~4 to express arbitrary Boolean
combinations of key requirements per door.

\subsection{Variation 6: Tools/Items (Multi-Door Keys)}
\label{subsec:onekey}

\subsubsection*{Specification}
Key $K_i$ unlocks a set $D(K_i) = \{D_{i_1}, D_{i_2}, \ldots\}$ of doors, so $\nK < n_D$ in general, with assignment $\sigma : \{1,\ldots,n_D\} \to \{1,\ldots,\nK\}$:
\begin{equation}
  \label{eq:phi_open}
  \varphi_{\mathrm{tool}} = \bigwedge_{j=1}^{n_D}
    \bigl(K_{\sigma(j)} \, \mathcal{R} \, \lnot D_j\bigr)
    \land \mathcal{F}(\mathcal T).
\end{equation}

\subsubsection*{Augmented graph modification}
The subset lattice is built over keys $\Kcal$ (not doors). The door-unlocking mapping becomes non-injective:
\[
  \mathbf D(S) = \bigcup_{K_i \in S} \mathbf D(K_i)
       = \{D_j : K_{\sigma(j)} \in S\}.
\]
When $K_i$ is added to $S$, all doors in $\mathbf D(K_i)$ are simultaneously
unlocked.

\begin{theorem}
  \label{thm:onekey}
  A shortest path from $s$ to $t$ in this augmented GCS corresponds
  to an optimal solution of~\eqref{optprob} with
  $\varphi_{\mathrm{tool}}$.
\end{theorem}
\begin{proof} See Appendix~\ref{app:variation_proofs}. \end{proof}

\begin{remark}
  When $\nK < \nD$, the augmented GCS has at most $2^{\nK}$ subgraphs rather than $2^{\nD}$, yielding an exponential reduction relative
  to a naive door-indexed construction.
\end{remark}

\begin{remark}[Reusable manipulation capabilities]
  \label{rem:tools}
  This variation admits a natural and rich interpretation beyond the literal
  key-door metaphor: a single ``key'' $K_i$ may represent a \textbf{tool} or \textbf{item}: a
  \emph{reusable manipulation capability} that reconfigures the
  environment at multiple locations.  A robot carrying a powered
  screwdriver can remove fasteners at any number of access panels;
  one equipped with a glass-breaking tool can breach multiple glazed
  barriers; a decontamination sprayer can neutralize hazardous zones
  at several points along a route. Similarly, a wire-cutting attachment
  can sever restraining cables at multiple junctions, a lubricant
  dispenser can free seized valves or hinges throughout a facility,
  and an RFID programmer can re-credential any number of
  electronically controlled gates sharing a common lock standard. In each case, acquiring the
  capability once  (visiting $K_i$) simultaneously unlocks all
  corresponding ``doors'' $\mathbf D(K_i)$, exactly as encoded by the
  non-injective mapping $\mathbf D(S) = \bigcup_{K_i \in S} \mathbf D(K_i)$.  The
  exponential reduction in subgraph count noted above (from $O(2^{n_D})$ to $O(2^{\nK})$
  when $\nK < n_D$) reflects the practical observation that a small
  number of distinct tool types may unlock a much larger number of
  individual obstacles throughout the environment.  This variation
  therefore provides a formal basis for manipulation planning
  problems in which the robot must reason jointly about which
  capabilities to acquire and in what order to deploy them across
  spatially distributed obstacles.
\end{remark}

\subsection{Variation 7: Timed and Bounded-Horizon Constraints}
\label{subsec:timed}

\subsubsection*{Specification}
Key $K_i$ must be collected and door $D_i$ passed within specified
time intervals:
\begin{equation}
  \label{eq:phi_time}
  \varphi_{\mathrm{time}} = \bigwedge_{i=1}^{\nK}
    \bigl(K_i \, \mathcal{R}_{[a^D_i, b^D_i]} \, \lnot D_i\bigr)
    \land \bigwedge_{i=1}^{\nK}
    \mathcal{F}_{[a^K_i, b^K_i]}(K_i) \land \mathcal{F}(\mathcal T).
\end{equation}

\subsubsection*{Augmented graph modification}
Timed constraints are encoded in convex sets rather than graph topology; the layer structure is unchanged.

\begin{construction}
  \label{con:timed}
  (i) Extend each vertex variable to include a time component:
  $x_v = (q_v, \tau_v)$ with $\tau_v \geq 0$.  (ii) For key node
  $K_i^{(S)}$, add convex constraint $\tau_{K_i} \in [a^K_i, b^K_i]$;
  for door node $D_i^{(S)}$, add $\tau_{D_i} \in [a^D_i, b^D_i]$.
  (iii) For each edge $(u,v) \in \Ehat$, add monotonicity constraint
  $\tau_v \geq \tau_u$.
\end{construction}

\begin{theorem}
  \label{thm:timed}
  A shortest path from $s$ to $t$ in Construction~\ref{con:timed}
  corresponds to an optimal solution of~\eqref{optprob} with
  $\varphi_{\mathrm{time}}$.
\end{theorem}
\begin{proof} See Appendix~\ref{app:variation_proofs}. \end{proof}

\textit{Worst-case subgraph count:} $O(2^{\nK})$.  All added constraints
(time windows, monotonicity) are linear and preserve the GCS convex
optimization structure.

\subsection{Variation 8: Conditional (If-Then) Specification}
\label{subsec:if-then}

\subsubsection*{Specification}
If the agent visits \emph{trigger region} $A_i$, then key-door precedence $K_i \,\mathcal{R}\, \neg D_i$ becomes active. Prior to visiting $A_i$, door $D_i$ is freely passable:
\begin{equation}\label{eq:cond_spec}
  \varphi_{\mathrm{cond}} = \bigwedge_{i=1}^{n}\Bigl(A_i \Rightarrow (K_i \,\mathcal{R}\, \neg D_i)\Bigr) \wedge \mathcal{F}(\mathcal{T}).
\end{equation}
This can be rewritten in pure STL as $\mathcal{G}(\neg A_i \vee (K_i \,\mathcal{R}\, \neg D_i))$, which enforces the precedence globally once the trigger is encountered.

\subsubsection*{Augmented graph modification}
\begin{construction}\label{con:conditional}
Index subgraphs by a joint state $(\mathbf{a}, S)$ where $\mathbf{a} \in \{0,1\}^n$ is a binary trigger vector and $S \subseteq \mathcal{K}$ is the collected key set. The total subgraph lattice has at most $2^n \cdot 2^{\nK} = 2^{n + \nK}$ subgraphs.

\textbf{Door-unlocking mapping:}
\begin{equation}\label{eq:cond_D}
  \mathbf{D}(\mathbf{a}, S) = \{D_i : a_i = 0\} \cup \{D_i : a_i = 1 \text{ and } K_i \in S\}.
\end{equation}
The first term captures doors whose precedence has not yet been triggered and are therefore freely passable. The second captures doors whose triggered precedence has been satisfied by collecting the corresponding key.

\textbf{Trigger transition edges:} For each trigger region $A_i$, add a directed edge at the trigger node from subgraph $\mathcal{G}_{(\mathbf{a}, S)}$ to $\mathcal{G}_{(\mathbf{a} \cup \{i\},\, S)}$ upon visiting $A_i$:
\[
  A_i^{(\mathbf{a},\, S)} \;\to\; A_i^{(\mathbf{a} \cup \{i\},\, S)}.
\]
This transition removes the freely-passable status of $D_i$ (setting $a_i = 1$), activating the precedence.

\textbf{Key transition edges:} Identical to the base construction, now additionally indexed by $\mathbf{a}$:
\[
  K_i^{(\mathbf{a},\, S \setminus \{K_i\})} \;\to\; K_i^{(\mathbf{a},\, S)}.
\]

\textbf{Edge sets:} For subgraph $\mathcal{G}_{(\mathbf{a},S)}$, the edge set is:
\[
  \mathcal E_{(\mathbf{a},S)} = \mathcal E_\emptyset \;\cup\; \{(u,v) \in \mathcal E \mid u \in \mathbf{D}(\mathbf{a},S) \text{ or } v \in \mathbf{D}(\mathbf{a},S)\}.
\]
All other aspects of Algorithm~2 are unchanged.
\end{construction}

\begin{theorem}\label{thm:conditional}
  A shortest path from $s$ to $t$ in the augmented GCS of Construction~\ref{con:conditional} corresponds to an optimal solution of problem~\eqref{optprob} with specification $\varphi_{\mathrm{cond}}$.
\end{theorem}
\begin{proof} See Appendix~\ref{app:variation_proofs}. \end{proof}

\begin{remark}
  The conditional specification subsumes the base release specification as the special case $A_i \equiv \top$ (trigger always active from the start), recovering $a_i = 1$ for all $i$ at $t = 0$.
\end{remark}

\begin{remark}
  Chained conditionals of the form $A_i \Rightarrow (B_i \Rightarrow (K_i \,\mathcal{R}\, \neg D_i))$ are encodable by composing trigger vectors: the state becomes $(\mathbf{a}^{(1)}, \mathbf{a}^{(2)}, S)$ with $|\mathbf{a}^{(1)}| = |\mathbf{a}^{(2)}| = n$, yielding subgraph count $O(4^n \cdot 2^{\nK})$. Each level of nesting multiplies the trigger component of the state space by a factor of 2, remaining singly exponential in total atom count.
\end{remark}

\begin{remark}
    Two sub-variations replace the release consequent with stronger obligations.
    The first uses an until consequent,
    \begin{equation}\label{eq:cond_spec_until}
      \varphi_{\mathrm{cond},\mathcal{U}} = \bigwedge_{i=1}^{n}\Bigl(A_i \Rightarrow (\neg D_i \,\mathcal{U}\, K_i)\Bigr) \wedge \mathcal{F}(\mathcal{T}),
    \end{equation}
    where the trigger incurs an \emph{obligation} to collect the key before the door may be passed.
    The second omits the door entirely,
    \begin{equation}\label{eq:cond_spec_future}
      \varphi_{\mathrm{cond},\mathcal{F}} = \bigwedge_{i=1}^{n}\Bigl(A_i \Rightarrow \mathcal{F}(K_i)\Bigr) \wedge \mathcal{F}(\mathcal{T}),
    \end{equation}
    requiring only that $K_i$ be visited before mission completion.
    Both sub-variations are naturally interpreted as a sensor-based event when the expected measurement is known at planning time: the trigger fires when a measurement at the current position exceeds a threshold, incurring an obligation to perform a subsequent action $K_i$.
    For example, a decontamination robot detecting radiation above a threshold at location $A_i$ incurs an obligation to visit a decontamination station $K_i$ before returning to base; or an inspection drone detecting a blade anomaly at turbine $A_i$ incurs an obligation to acquire a confirmatory close-range observation at $K_i$. When the trigger condition is unknown until execution, the obligation $\mathcal{F}(K_i)$ can be injected dynamically via event-triggered replanning.
\end{remark}

\begin{remark}[Combining and Composing Variations]
The variations are largely orthogonal and can be combined for form highly expressive mission specifications. Timed constraints compose cleanly with all other variations because
  they modify only convex sets, not graph topology.  For example,
  combining with Variation~2 (ordered keys) requires only appending
  time windows to the linear chain node sets.
\end{remark}

\subsection{Fragment Admissibility and the Double-Exponential Barrier}
\label{subsec:fragment}

The computational advantage of the augmented GCS over general-purpose temporal logic motion planning tools stems from a structural property of the specification fragments it targets.  The classical pipeline for arbitrary STL or LTL motion planning faces two sequential exponential blowups: constructing a finite automaton from the formula costs $2^{O(|\varphi|)}$ in formula length, and forming the product of that automaton with the planning state space multiplies the exponential again, yielding double-exponential complexity in the worst case \cite{baier2008principles}.  The augmented GCS avoids the double blowup by building the subgraph structure directly from the semantics of the formula, exploiting the fact that each specification variation in Section~\ref{sec:variations} admits a \emph{progress monitor} (a deterministic finite-state tracker of which actions have been completed and which regions are currently accessible) whose state space is at most singly exponential in the number of logical atoms. The augmented GCS subgraph lattice is precisely this progress monitor embedded in the continuous geometry, giving a construction that is singly exponential in $\nK$ by design.

This \emph{bottom-up} philosophy --- identifying specific fragments whose structure can be exploited directly, rather than applying a top-down pipeline to arbitrary formulas --- is the organizing principle behind the specification library developed here.  Each variation in Section~\ref{sec:variations} represents a fragment for which the progress monitor is known, the modification to the augmented GCS construction is explicit and tractable, and the correctness proof is tractable. The natural question is how far this library can be extended: which additional STL fragments admit singly exponential progress monitors and therefore tractable augmented GCS encodings, and where does the double-exponential barrier become unavoidable?  Precisely characterizing this \emph{admissibility boundary}, including a formal definition of admissible fragments, a hierarchy organizing them by the algebraic structure of their progress states, and syntactic sufficient conditions for admissibility, is the subject of ongoing work and will be developed in a companion paper expanding the specification library.

\section{Numerical Experiments}

In this section, we demonstrate and evaluate the augmented GCS framework for solving two key-door benchmarks from \cite{kurtz2023temporal}, procedurally generated key-door mazes, and a hand-designed key-door environment from a classic video game. Experiments were performed on a laptop with an Apple M2 CPU and 8GB RAM, unless otherwise noted. The underlying convex optimization solver was MOSEK \cite{mosek}, accessed via the GCS components from the Drake \cite{drake} Python bindings. The relaxed optimality gap is reported as
\[
  \delta_{\mathrm{relax}} = \frac{C_{\mathrm{round}} -
    C_{\mathrm{relax}}}{C_{\mathrm{relax}}},
\]
where $C_{\mathrm{relax}}$ lower-bounds the optimal value and
$C_{\mathrm{round}}$ is the rounded solution cost.  We use the rounding implementation from \cite{marcucci2023motion} with 100 trials and 10 max paths. All experiments minimize path length only, with $\alpha = 1$ and $\beta=0$ in problem \eqref{optprob}. The source code is
publicly available at
\path{https://github.com/TSummersLab/augmented-gcs.git}.

\subsection{Key-Door Environments from \cite{kurtz2023temporal}}
\textbf{2-Key Environment.} We first consider the simple key-door environment with two keys and two doors shown in Fig. \ref{fig:simple key door}. The robot must pick up two keys to open two doors in order to reach the target set, which can be written as the STL formula $\varphi= \mathcal{K}_1 \mathcal{R} \neg \mathcal{D}_1 \wedge \mathcal{K}_2 \mathcal{R} \neg \mathcal{D}_2 \wedge \mathcal{F} \, \mathcal{T}$.

We created an exact convex partition and labeled GCS by inspection and used Algorithm \ref{Algo1} to construct the augmented GCS, shown in Fig. \ref{fig:augGCS}. The augmented GCS construction time is \textbf{0.0108} seconds and the solving time is \textbf{0.249} seconds, certified globally optimal with $\delta_{\mathrm{relax}}=0$. Compared with an approach that uses general-purpose temporal logic tools \cite{kurtz2023temporal}, our augmented GCS construction is about 80$\times$ faster. 

\textbf{5-Key Environment.}
We also solved a 5 key-door environment from \cite{kurtz2023temporal}. The STL formula can be written as
\begin{equation}
\begin{split}
    \varphi = \bigwedge_{i=1}^{5}
    \bigl(K_i \, \mathcal{R} \, \lnot D_i\bigr) \land \mathcal{F}(\mathcal{T}).
\end{split}
\end{equation}
We extracted an exact convex partition by inspection, which produces a graph with 7 free space cells, 5 key cells, 5 door cells. The augmented GCS constructed by Algorithm \ref{Algo1} contains six layers with nine subgraphs.  The augmented GCS construction time is \textbf{0.0423} seconds and the solving time is \textbf{1.573} seconds, certified globally optimal with $\delta_{\mathrm{relax}}=0$. Compared with an approach that uses general-purpose temporal logic tools \cite{kurtz2023temporal}, our augmented GCS construction is more than 50000$\times$ faster, and our solve time is about 3.5$\times$ faster.

More generally, compared with the top-down, full-grammar pipeline targeted by \cite{kurtz2023temporal}, our bottom-up approach is \emph{exponentially} faster by exploiting the structure in the specification fragment, which will compound further with problem size. This is due to the fundamental doubly exponential complexity of converting an arbitrary temporal logic specification into a finite automaton versus the singly exponential complexity of our fragments, matching BHK. 

\begin{table*}
\centering
\caption{Numerical Experiments with Key-Door Mazes}
\begin{tabular}{>{\centering\arraybackslash}m{0.03\textwidth}
                >{\centering\arraybackslash}m{0.07\textwidth}
                >{\centering\arraybackslash}m{0.05\textwidth}
                >{\centering\arraybackslash}m{0.05\textwidth}
                >{\centering\arraybackslash}m{0.04\textwidth}
                >{\centering\arraybackslash}m{0.04\textwidth}
                >{\centering\arraybackslash}m{0.07\textwidth}
                >{\centering\arraybackslash}m{0.07\textwidth}
                >{\centering\arraybackslash}m{0.09\textwidth}
                >{\centering\arraybackslash}m{0.09\textwidth}
                >{\centering\arraybackslash}m{0.06\textwidth}}
\toprule
ID & Maze Size & $|\mathcal{V}|$ & $|\mathcal{E}|$ & $\#$ of keys & Max Width & $|\hat{\mathcal{V}}|$ & $|\hat{\mathcal{E}}|$ & Form GCS (sec) & Solve (sec) & $\delta_{relax}$ $(\%)$ \\
\midrule
1 & \multirow{3}{*}{9 $\times$ 9} & 18 & 34 & 1 & 1 & 37 & 68 & 0.00433 & 0.0396 & 0 \\
2 &  & 21 & 40 & 2 & 2 & 85 & 154 & 0.00561 & 0.0495 & 0 \\
3 &  & 23 & 44 & 3 & 3 & 185 & 330 & 0.0123 & 0.132 & 0 \\
\midrule
4 & \multirow{5}{*}{19 $\times$ 19} & 63 & 124 & 1 & 1 & 127 & 248 & 0.00763 & 0.0566 & 0 \\
5 &  & 66 & 130 & 2 & 1 & 199 & 384 & 0.0119 & 0.0695 & 0 \\
6 &  & 72 & 142 & 3 & 2 & 361 & 690 & 0.0225 & 0.166 & 0 \\
7 &  & 74 & 146 & 5 & 10 & 2369 & 4514 & 0.169 & 2.71 & 0 \\
8 &  & 69 & 136 & 4 & 6 & 1105 & 2114 & 0.0764 & 1.71 & 1.06 \\
\midrule
9 & \multirow{8}{*}{29 $\times$ 29} & 129 & 256 & 1 & 1 & 259 & 512 & 0.0152 & 0.0975 & 0 \\
10 &  & 145 & 288 & 3 & 3 & 1161 & 2282 & 0.0749 & 1 & 0 \\
11 &  & 149 & 296 & 6 & 2 & 1193 & 2294 & 0.0711 & 0.652 & 0 \\
12 &  & 149 & 320 & 6 & 14 & 7153 & 15104 & 0.568 & 35.8 & 0 \\
13 &  & 149 & 320 & 6 & 10 & 5365 & 11280 & 0.415 & 26.1 & 1.97 \\
14 &  & 154 & 306 & 10 & 2 & 2157 & 4042 & 0.128 & 3.74 & 0 \\
15 &  & 163 & 330 & 10 & 6 & 6358 & 12362 & 0.439 & 29.6 & 0.106 \\
16 &  & 162 & 328 & 10 & 20 & 12799 & 24458 & 0.981 & 106 & 1.72 \\
\midrule
17 & \multirow{8}{*}{39 $\times$ 39} & 242 & 482 & 1 & 1 & 485 & 964 & 0.0297 & 0.184 & 0 \\
18 &  & 246 & 490 & 5 & 1 & 1477 & 2892 & 0.0867 & 0.613 & 0 \\
19 &  & 251 & 500 & 5 & 10 & 8033 & 15842 & 0.612 & 45.5 & 0.282 \\
20 &  & 243 & 484 & 7 & 2 & 2188 & 4252 & 0.129 & 1.41 & 0 \\
21 &  & 249 & 496 & 7 & 35 & 31873 & 62594 & 2.7 & 597 & 1.93 \\
22 &  & 255 & 508 & 10 & 3 & 5101 & 9838 & 0.325 & 18.9 & 0.915 \\
23 &  & 264 & 536 & 10 & 6 & 8185 & 16282 & 0.584 & 48.8 & 0 \\
24 &  & 264 & 552 & 10 & 35 & 37489 & 75922 & 3.22 & 871 & 1.1 \\
\midrule
25 & \multirow{2}{*}{99 $\times$ 99} & 1458 & 2914 & 1 & 1 & 2917 & 5828 & 0.182 & 1.45 & 0 \\
26 &  & 1490 & 2978 & 3 & 2 & 7451 & 14874 & 0.468 & 9.64 & 0 \\
\midrule
27 & \multirow{3}{*}{199 $\times$ 199} & 5946 & 11890 & 3 & 1 & 23785 & 47544 & 1.45 & 41.2 & 0 \\
28 &  & 5988 & 11974 & 3 & 3 & 47905 & 95770 & 3.28 & 552 & 0 \\
29 &  & 5992 & 11982 & 3 & 2 & 29961 & 59890 & 1.87 & 119 & 0 \\
\bottomrule
\end{tabular}
\label{tab:experiments results}
\end{table*}

\begin{figure}
  \centering
  \begin{subfigure}[b]{0.45\linewidth}
    \centering
    \includegraphics[width=\linewidth]{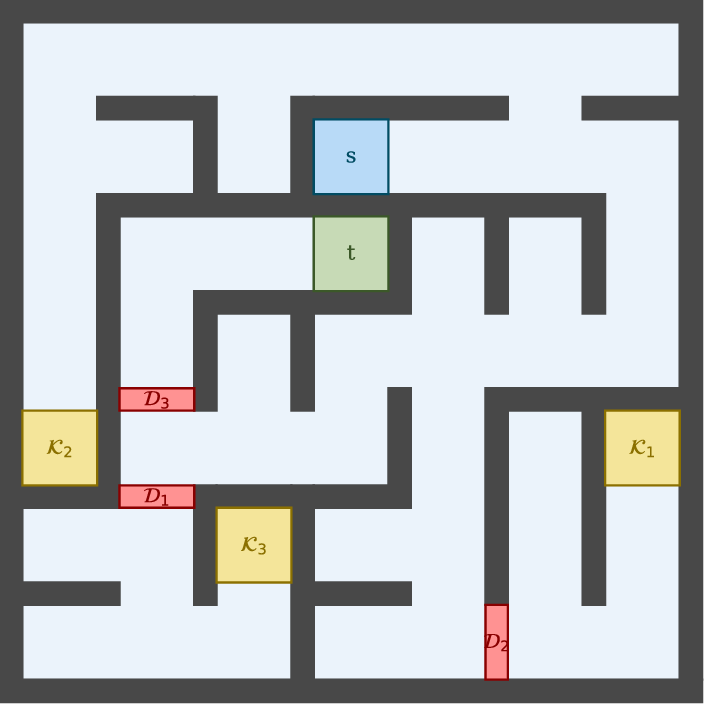}
    \caption{}
    \label{fig:maze_for_graph}
  \end{subfigure}
  \hfill
  \begin{subfigure}[b]{0.45\linewidth}
    \centering
    \includegraphics[width=\linewidth]{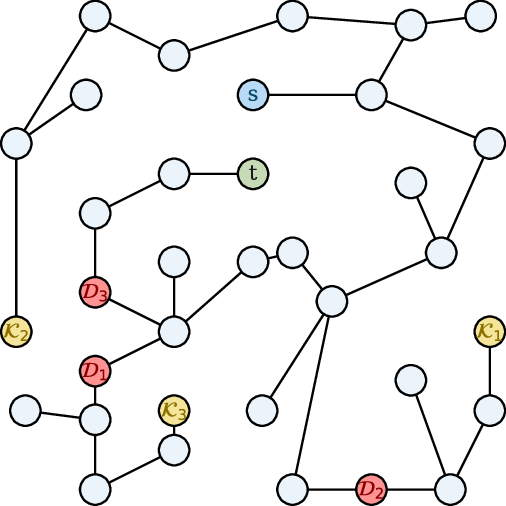}
    \caption{}
    \label{fig:graph_of_maze}
  \end{subfigure}
  \caption{(a) key-door maze, (b) graph of maze.}
  \label{fig:maze_graph_representation}
\end{figure}

\subsection{Numerical Experiments on Key-Door Mazes}
To further evaluate the augmented GCS framework, we developed a maze benchmark generator (described in
Appendix~\ref{app:maze_generator}) that produces key-door maze environments
of controllable size and width. Fig. \ref{fig:maze_graph_representation} shows a partitioned and merged maze and its associated graph.




We evaluate the performance of the proposed approach on a set of key-door mazes created by the our generator. The convex free space partition and labeled GCS are constructed from the maze generation process without requiring Algorithm \ref{alg:partition}. From the labeled GCS for each maze, we used Algorithm \ref{Algo1} to construct the augmented GCS and then solved a shortest path problem on the augmented GCS based on \cite{marcucci2023motion,marcucci2024shortest} using Drake and Mosek. Table \ref{tab:experiments results} shows the results of 29 experiments. For each experiment, we report the maze size; number of vertices in the base GCS $|\mathcal{V}|$; number of edges in the base GCS $|\mathcal{E}|$; number of keys $\nK$; maximum width of the augmented GCS; number of vertices in the augmented GCS $|\hat{\mathcal{V}}|$; number of edges in the augmented GCS $|\hat{\mathcal{E}}|$; the augmented GCS construction time; the augmented GCS solving time; and an upper bound on the optimality gap $\delta_{\mathrm{relax}}$. Fig.~6 shows a 3-key, 3-door, 99$\times$99 maze, which we solve we solve to global optimality within $1.5s$.

\begin{figure}
  \centering
  \includegraphics[width=0.85\linewidth]{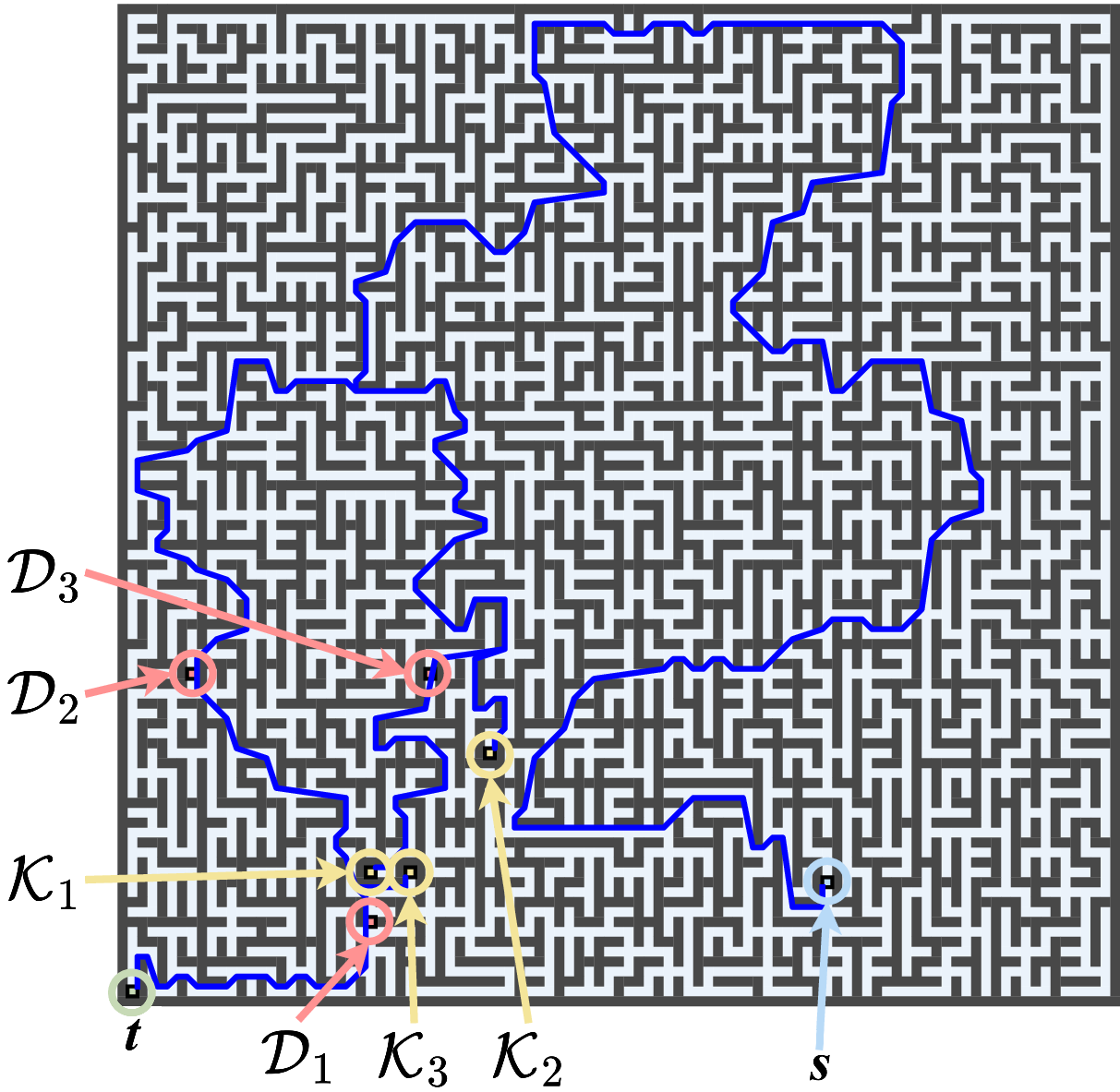}
  \caption{A 99$\times$99 maze with 3 keys and 3 doors.}
  \label{fig:large maze}
\end{figure}

\textbf{Discussion.}
We make several observations based on the results in Table \ref{tab:experiments results}. First, our approach scales to far larger problems than reported in \cite{kurtz2023temporal}, obtaining optimal or near-optimal solutions for up to $6$k-node base graphs, up to 10 keys, and up to an augmented GCS max width of 35. The augmented GCS construction times are within a few seconds across all instances, and the solve times scale mainly with the number of augmented GCS nodes and edges, which strongly depend on the maximum width and reflect the complexity of mission logic; this is illustrated graphically in Fig.~\ref{fig:scaling plot}. The optimality gap is certified to within 2\% of the global optimum across all instances, with many cases globally optimal.

\begin{figure}
  \centering
  \hspace*{-0.8cm}
  \includegraphics[width=0.9\linewidth]{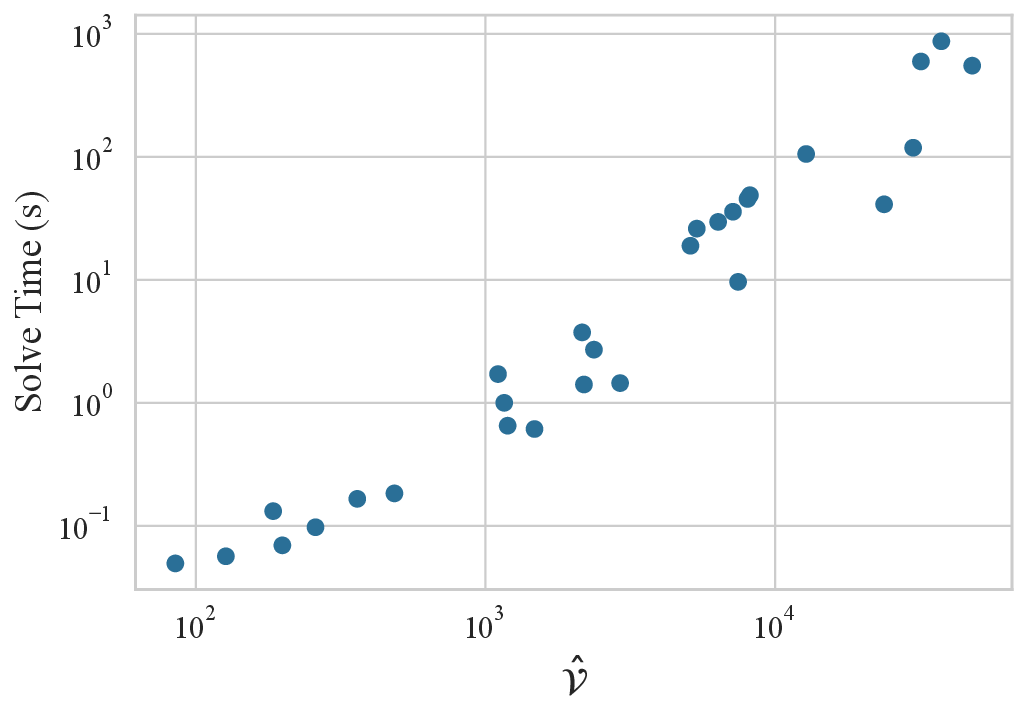}
  \caption{Log-log plot showing singly exponential growth of solve time vs. augmented GCS node count, which depends strongly on max width.} 
  \label{fig:scaling plot}
\end{figure}

\subsection{Pure Wayset Collection (TSP) Experiments}
\label{subsec:tsp_exp}
We now illustrate Variation 2 (pure wayset collection (TSP)) from the specification library in Section \ref{sec:variations}. Here, the augmented GCS produces the full $2^{\nK}$ BHK subgraph lattice with maximum width, which allows us to push Algorithm \ref{Algo1} to its limit. These experiments were performed on a laptop with an Intel Core Ultra 7 258V and 32 GB of RAM.

We generated random instances of region-TSP with $\nK \in \{3,5,7,9,11\}$ waysets and no obstacles. Each instance uniformly places $\nK$ center points in $[0, 1]^2$ and then creates polytopic waysets from a convex hull of random points around the center. Since there are no obstacles, each subgraph is a complete graph on key nodes with $\nK(\nK-1)$ edges. 
For each instance the shortest path problem on the augmented GCS is solved using Drake \cite{drake} and Mosek \cite{mosek}, using the first key vertex in $\Vcal$ as both the start and target. Table \ref{tab:tsp_exp} shows the results from the 500 experiments, 100 per $\nK \in \{3,5,7,9,11\}$. For each $\nK$ we report the number of vertices and edges in the augmented graph, the augmented GCS construction time, and the mean and medium solve time and optimality gap. Figure \ref{fig:tsp_11} shows one instance of $\nK=11$, with certified globally optimal trajectory.


\begin{table}[t]
  \caption{Pure Waypoint Collection (TSP) Experiments}
  \label{tab:tsp_exp}
  \centering
  \small
  \setlength{\tabcolsep}{3pt}
  \begin{tabular}{cccccccc}
    \toprule
    $\nK$ & $|\hat\Vcal|$ & $|\hat\Ecal|$ 
    & Form (s) 
    & \multicolumn{2}{c}{Solve (s)} 
    & \multicolumn{2}{c}{$\delta_\text{relax}$ (\%)} \\
    \cmidrule(lr){5-6} \cmidrule(lr){7-8}
    & & 
    & 
    & mean & median 
    & mean & median \\
    \midrule
    3 & 12 & 28 & 0.001 & 0.016 & 0.016 & 0.000 & 0.000 \\
    5 & 80 & 352 & 0.008 & 0.204 & 0.200 & 0.117 & 0.000 \\
    7 & 448 & 2880 & 0.057 & 2.194 & 2.116 & 0.005 & 0.000 \\
    9 & 2304 & 19456 & 0.395 & 31.17 & 28.04 & 0.113 & 0.000 \\
    11 & 11264 & 117760 & 4.172 & 191.6 & 173.7 & 0.047 & 0.000 \\
    \bottomrule
  \end{tabular}
\end{table}



\begin{figure}
  \centering
  \includegraphics[width=0.85\linewidth]{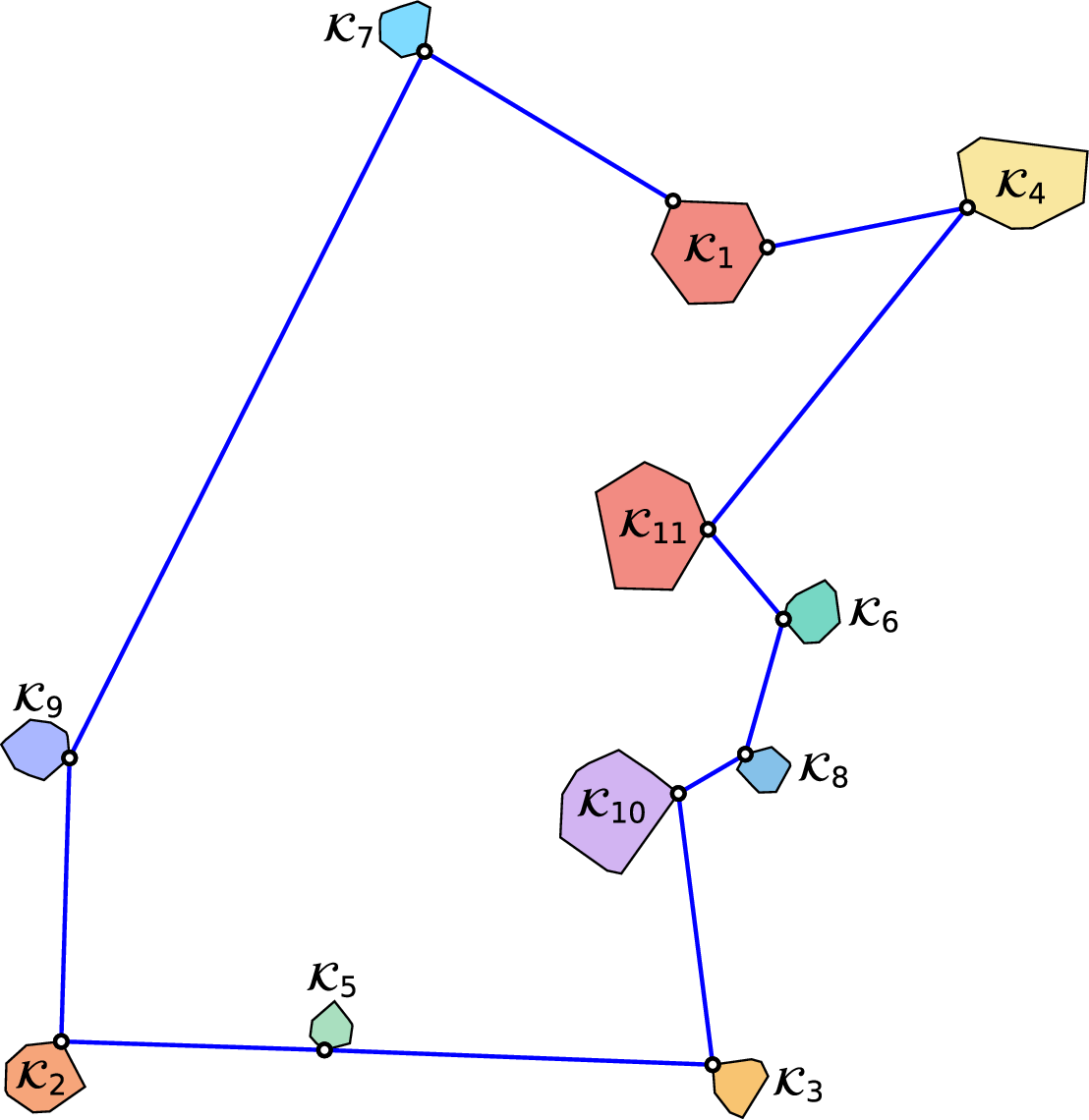}
  \caption{Optimal wayset solution for $\nK=11$.}
  \label{fig:tsp_11}
\end{figure}

\textbf{Discussion.} Table \ref{tab:tsp_exp} shows singly exponential scaling with $\nK$ for both augmented GCS size and solve times, matching the BHK algorithm. The optimality gaps show that most instances are certified globally optimal. Further analysis reveals that non-zero optimality gap is typically due to the relaxation, not the rounding. 

Although the singly exponential scaling of our approach is far better than the doubly exponential scaling of top-down approaches for arbitrary STL formulas, these results strongly motivate the use of heuristics and pruning strategies that trade solution quality for speed. The BHK correspondence suggests that 70 years of research on TSP approximation algorithms may provide inspiration for scalable
approximation schemes with suboptimality bounds.
An exploration of how such heuristics could be translated to the augmented GCS framework was initiated in \cite{luna2026augmented}, where high quality solutions can be obtained for up to $\nK = 100$ waysets within one minute solve times.

\subsection{Hand-Designed Maze: Legend of Zelda Eagle Dungeon}
\label{subsec:zelda}
 
Procedurally generated mazes provide statistical coverage of the
augmented GCS parameter space, but their
key-door dependencies are placed algorithmically rather than
designed by a human expert.  To complement the maze benchmark, we
consider the Eagle Dungeon (Level 1) from \emph{The Legend of
Zelda} (Nintendo, NES, 1986. Designers: Shigeru Miyamoto and Takashi Tezuka), a canonical example of a
human-designed environment in which sequential access constraints
arise intentionally from the designer's puzzle logic.  Video game
dungeons are a particularly apt benchmark
class for key-door planning frameworks, since game designers invest
substantial effort calibrating precedence structures to be
non-trivial but tractable, a regime where
algorithmic methods should demonstrate clear value.
 
\subsubsection*{Environment description}
The Eagle Dungeon, shown in Fig.~\ref{fig:eagle dungeon}, consists of a grid of interconnected rooms arranged
on a $16 \times 16$ tile map.  Each room contains obstacles made of one or more axis-aligned tiles. 
We manually extracted the half-space representations of all obstacles and obtained a polytopic
partition of the free space using
Algorithm~\ref{alg:partition} with greedy merging \cite{shaikh2025exact}. The dungeon contains $\nK = 6$ small keys,
$n_D = 6$ locked doors; the terminal goal is the Triforce room at the dungeon's end (top right).
The start position is the dungeon entrance at the bottom. The resulting labeled graph, shown in Fig.~\ref{fig:zelda dungeon GCS}, has $|\mathcal V| = 153$
vertices and $|\mathcal E| = 440$ edges. The dungeon's room-graph and obstacle structure produces a labeled
GCS with a substantial number of free-space vertices relative to the
number of logical constraints, placing it in a qualitatively
different regime from the mazes of
Table~\ref{tab:experiments results}. 

\begin{figure}
  \centering
  \includegraphics[width=0.95\linewidth]{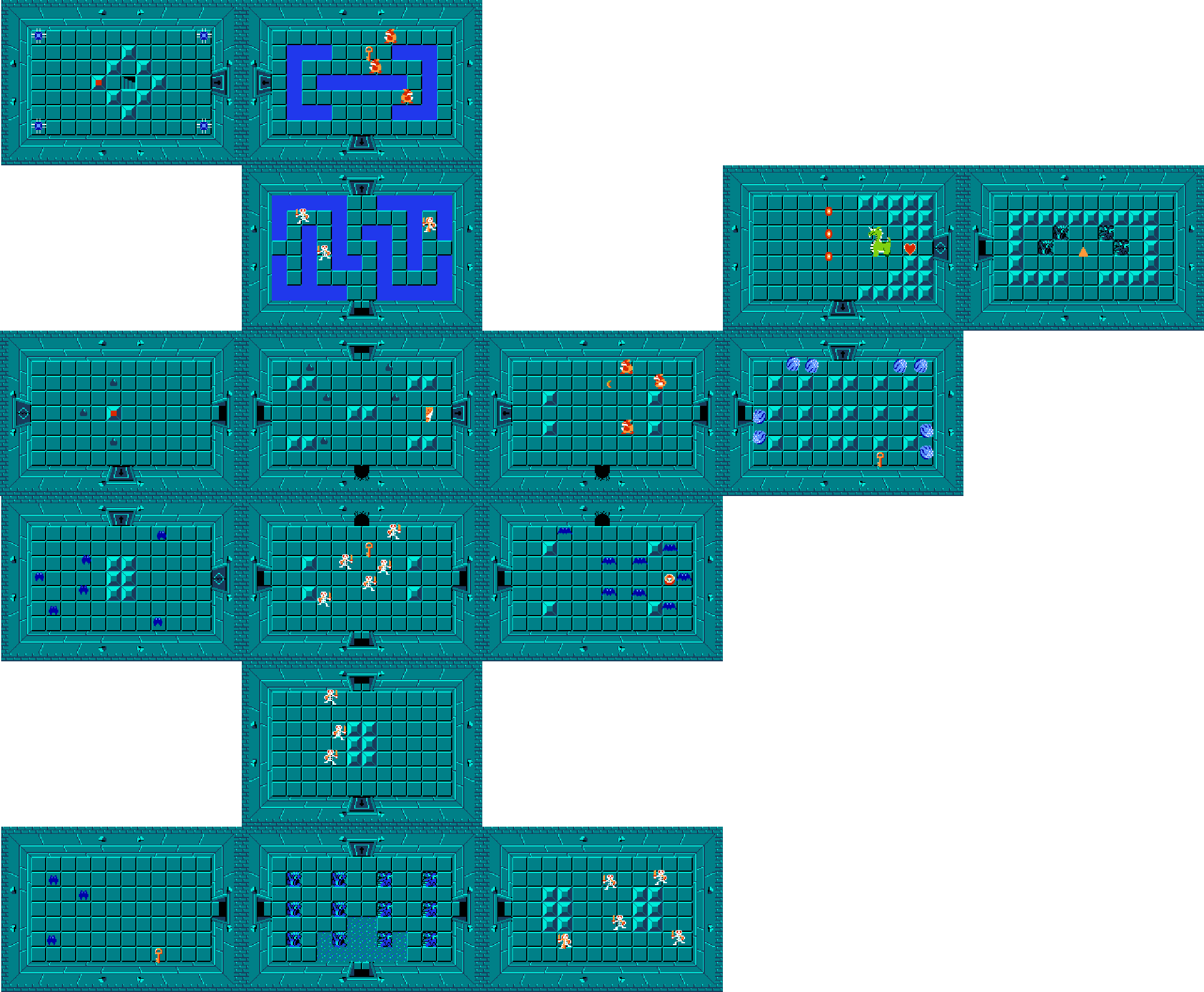}
  \caption{\emph{The Legend of Zelda} Eagle Dungeon (Nintendo, 1986).}
  \label{fig:eagle dungeon}
\end{figure}

\begin{figure}
  \centering
  \includegraphics[width=0.95\linewidth]{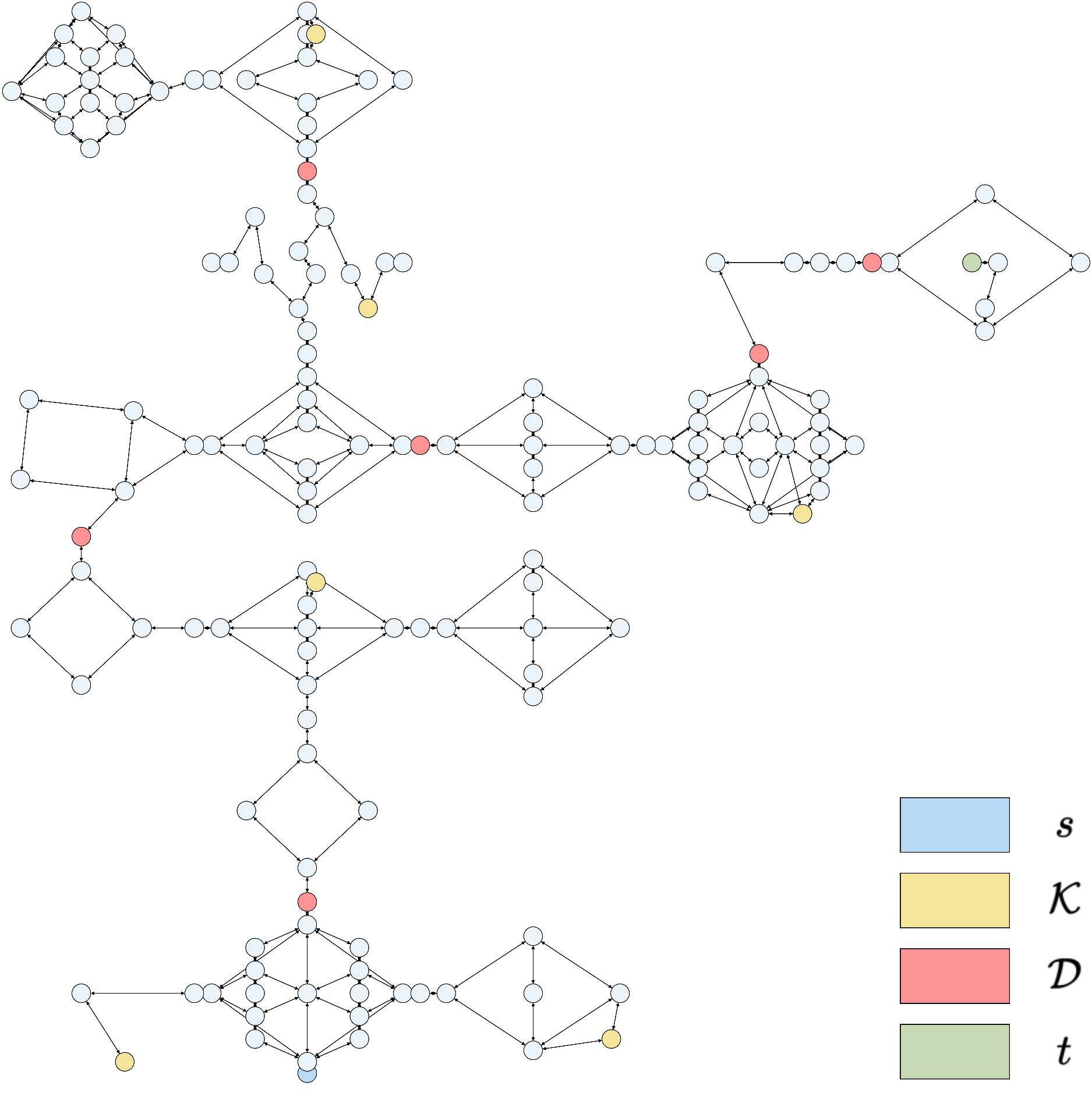}
  \caption{Base GCS of the Eagle Dungeon.}
  \label{fig:zelda dungeon GCS}
\end{figure}

\begin{figure}
  \centering
  \includegraphics[width=0.7\linewidth]{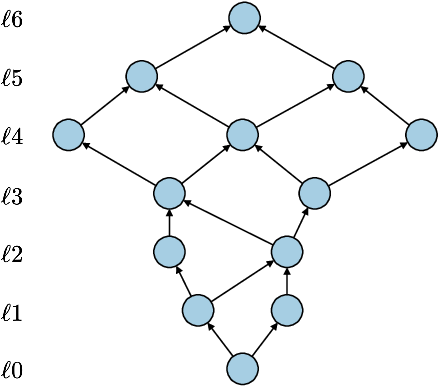}
  \caption{Augmented GCS of the Eagle Dungeon.}
  \label{fig:zelda dungeon AGCS}
\end{figure}

\begin{figure}
  \centering
  \includegraphics[width=0.95\linewidth]{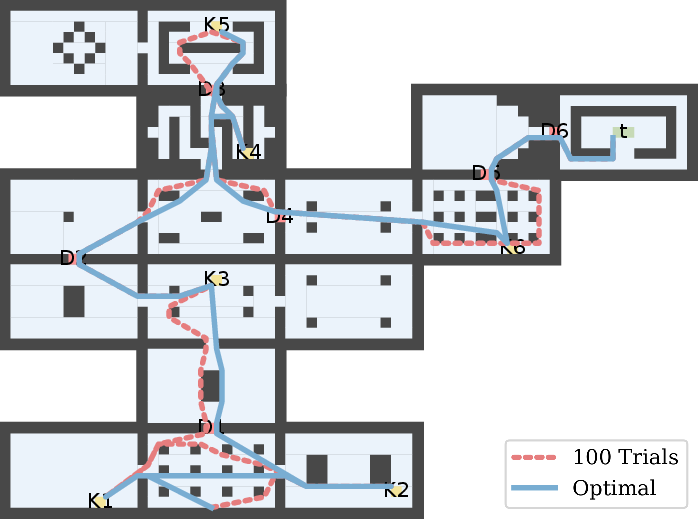}
  \caption{The solid blue line shows the global optimal solution. The dashed red line shows a feasible solution obtained from 100 rounded trials, with relaxation  gap 17.9\% and true optimality gap 9.53\%.}
  \label{fig:zelda dungeon linear solution}
\end{figure}

\begin{table}[t]
\centering
\caption{Zelda Dungeon Experiments}
\label{tab:zelda_experiments}
\renewcommand{\arraystretch}{1.3}
\setlength{\tabcolsep}{8pt}
\begin{tabular}{cccc}
\toprule
Rounded Trials & 
Solve (s) & 
$\delta_{\text{relax}}$ (\%) &
$\delta$ (\%) \\
\midrule
10   & 10.26  & 24.2 & 15.1\\
100  & 73.61  & 17.9  & 9.53\\
1000 & 713.2  & 11.9  & 3.67\\
\bottomrule
\end{tabular}
\end{table}
 
 
\subsubsection*{Specification}
 
The small keys in the Eagle Dungeon are \emph{fungible} and \emph{consumable}: any small
key opens any locked small-key door within the dungeon, and a single small key is consumed upon passage through
any locked door and cannot be reused.  For the
purposes of this experiment we adopt a simplified specification that
assigns each small key a fixed identity (key $K_i$ unlocks door
$D_i$)\footnote{We also omit the dungeon's key item (in this case, the bow) and move one locked door from the top left towards the item room to the top right towards the triforce room. In the original game, the Eagle Dungeon can be completed without acquiring the bow.}, which is an instance of the base release specification
$\varphi_{\mathrm{base}}$ from~\eqref{eq:phi_base}:
\begin{equation}
  \label{eq:phi_zelda}
  \varphi_{\mathrm{Eagle}} =
    \bigwedge_{i=1}^{6}
    \bigl(K_i \, \mathcal{R} \, \lnot D_i\bigr)
    \land \mathcal{F}(\text{triforce}).
\end{equation}
The full Zelda mechanic is a richer
specification than~\eqref{eq:phi_zelda} and lies outside the
variation library of the current paper.  The simplified
specification~\eqref{eq:phi_zelda} therefore represents a relaxation
of the true game logic in which key identities are fixed rather than
fungible.  We discuss the implications of this gap and the path
toward closing it below.
The augmented GCS associated with this specification is obtained from Algorithm \ref{Algo1}, has max width 3, and is shown in Fig.~\ref{fig:zelda dungeon AGCS}.
Solving a convex relaxation of the shortest path problem on the augmented GCS with 100 rounding trials in about 10 seconds yields the feasible trajectory in Fig.~\ref{fig:zelda dungeon linear solution} with relaxation gap $\delta_{\mathrm{relax}} = $ 17.9\%. Also shown is a globally optimal path, obtained by solving the exact mixed-integer convex reformulation of the shortest path problem on the augmented GCS, requiring about 6 hours.
 
\subsubsection*{Discussion}
 
The convex relaxation exhibits a relatively large relaxation gap on the Eagle Dungeon compared to procedurally generated maze instances. We attribute this mainly to the environment geometry. In particular, certain rooms contain numerous small obstacles, resulting in a fine partition of the free space into many small interconnected convex cells. This induces many distinct paths between room entrances and exits. In the GCS relaxation, flow is distributed fractionally across these alternative routes, and the rounding procedure recovers a discrete path by sampling edges according to their fractional weights. With a limited number of trials, the algorithm explores only a small portion of this combinatorial path space, often yielding suboptimal or duplicated paths. Increasing the number of rounding trials (e.g., from 100 to 1000) improves coverage of candidate paths and reduces the optimality gap to $11.9\%$, at the cost of substantially increased computation time; see Table \ref{tab:zelda_experiments}.

Furthermore, we observe that the relaxation itself is relatively loose in this instance. Solving the mixed-integer program yields a globally optimal cost of $192.8$, while the relaxation cost is $178.6$. Thus, it is worth emphasizing that the relaxation gap may be attributed to either suboptimality or poor quality of the lower bound (or both). This highlights that the observed optimality gap is influenced not only by the effectiveness of the rounding procedure, but also by the tightness of the convex relaxation, both of which are adversely affected by the highly fragmented structure of the environment.

Investigating relaxation tightness on human-designed environments with irregular geometry (as opposed to procedurally generated mazes) is an interesting direction for future work, and motivates developing tighter mixed-integer formulations or problem-specific rounding strategies for this instance class.
 
\subsubsection*{Specification gap and future directions}
 
The Eagle Dungeon highlights a natural limit of the current
specification library: the fungible-consumable-key mechanic, in which a
single key is expended upon door passage and cannot be reused,
is not captured by any of the eight variations in
Section~\ref{sec:variations}.  Extending the augmented GCS
framework to fungible and consumable keys, and more generally developing the
full specification library needed to handle the rich precedence
structures found in human-designed environments, is a natural
direction for future work.
 
More broadly, human-designed environments (game dungeons,
building floor plans with access-controlled zones, multi-stage
inspection routes) represent a benchmark class qualitatively
distinct from procedurally generated instances.  Their key-door
structures are crafted to be challenging yet solvable, their
geometry is irregular and purposeful, and their optimal solutions
carry semantic meaning that procedural benchmarks lack.  We view
the Eagle Dungeon as an initial representative of this class and
anticipate that a systematic benchmark library drawn from
hand-designed environments will be a productive avenue for
evaluating and extending the augmented GCS framework.

\section{Conclusion}
\label{sec:conclusion}

We presented a framework for motion planning with logical precedence
specifications based on augmented graphs of convex sets.  The
augmented GCS encodes key-door precedence constraints exactly in the graph
topology, enabling the GCS shortest-path machinery to simultaneously
optimize key-collection sequencing and continuous trajectory geometry.
We provided formal proofs of optimality, established a precise
correspondence with the Bellman--Held--Karp dynamic programming
algorithm that reveals the augmented GCS as a continuous-geometry
generalization of the shortest Hamiltonian path and TSP problems, and
developed a library of eight systematic variations based on STL fragments that admit
tractable augmented GCS encodings.  Extensive numerical experiments
confirmed exponential speedup over general-purpose
temporal logic tools and solution quality within 2\% of globally
optimal across almost all tested instances.

Several limitations of the current framework motivate future work.
Scaling to environments with large numbers of keys and high augmented
GCS width will likely require heuristics or approximations that
trade off computation time with solution quality; the BHK
correspondence suggests that TSP approximation algorithms may provide inspiration for systematic
approximation schemes with suboptimality bounds. The convex partition
quality (number of free-space nodes) significantly affects solve
time, and adaptive or online partition refinement strategies would
be valuable.  Extending the framework to dynamic environments,
stochastic settings, and multi-agent problems, where the
key-door precedence structure and GCS convexity can still be
exploited, are natural directions.  Finally, integrating the
offline augmented GCS planning with real-time feedback control and
online replanning for environments with uncertainty and disturbances
remains an important open challenge.

\section*{Acknowledgements}
\label{sec:acknowledgements}
The authors gratefully acknowledge the contributions of J.~Shaikh, D.~Gostin, Y.~Xiang, and J.~Koeln to the base framework in \cite{you2025motion}.

During the final phase of manuscript prepration, following the completion of core theoretical development and numerical experiments, the authors utilized Claude Sonnet 4.6 (Anthropic) as a tool for high-level conceptual exploration, and additionally for prose refinement. This workflow was built on highly specific, non-trivial technical prompts grounded in the rich context of the authors' prior work and expertise across optimization, motion planning, dynamics, and feedback control, and especially the corresponding author's extensive research experience and deep domain expertise. Critically, the AI served as an amplifier and concentrating lens of focused human expertise, functioning not only as a force multiplier but also as a kind of intellectual microscope: enabling deeper, faster navigation of the theoretical landscape without substituting for the authors' own reasoning. All AI-facilitated suggestions underwent rigorous manual interpretation and independent verification by all authors. In particular, all mathematical formulations, theoretical results, and proofs originated from and were validated by the human authors; all core algorithmic logic in the numerical experiments was handcrafted by the authors. The authors assume full intellectual responsibility for all claims in this work.

\bibliographystyle{IEEEtran}
\bibliography{bibliography}

@article{marcucci2023motion,
  title={Motion planning around obstacles with convex optimization},
  author={Marcucci, Tobia and Petersen, Mark and Von Wrangel, David and Tedrake, Russ},
  journal={Science robotics},
  volume={8},
  number={84},
  pages={eadf7843},
  year={2023},
  publisher={American Association for the Advancement of Science}
}

@article{marcucci2024shortest,
  title={Shortest paths in graphs of convex sets},
  author={Marcucci, Tobia and Umenberger, Jack and Parrilo, Pablo and Tedrake, Russ},
  journal={SIAM Journal on Optimization},
  volume={34},
  number={1},
  pages={507--532},
  year={2024},
  publisher={SIAM}
}

@article{marcucci2025unified,
  title={A Unified and Scalable Method for Optimization over Graphs of Convex Sets},
  author={Marcucci, Tobia},
  journal={arXiv preprint arXiv:2510.20184},
  year={2025}
}

@article{kurtz2023temporal,
  title={Temporal logic motion planning with convex optimization via graphs of convex sets},
  author={Kurtz, Vince and Lin, Hai},
  journal={IEEE Transactions on Robotics},
  volume={39},
  number={5},
  pages={3791--3804},
  year={2023},
  publisher={IEEE}
}

@article{you2025motion,
  title={Motion Planning with Precedence Specifications via Augmented Graphs of Convex Sets},
  author={You, Shilin and Luna, Gael and Shaikh, Juned and Gostin, David and Xiang, Yu and Koeln, Justin and Summers, Tyler},
  journal={arXiv preprint arXiv:2510.22015},
  year={2025}
}

@book{belta2017formal,
  title={Formal methods for discrete-time dynamical systems},
  author={Belta, Calin and Yordanov, Boyan and Gol, Ebru Aydin},
  volume={89},
  year={2017},
  publisher={Springer}
}

@inproceedings{shaikh2025exact,
  title={Exact Obstacle-Free Space Representation Using Hybrid Zonotopes},
  author={Shaikh, Juned and Gostin, David and Koeln, Justin P},
  booktitle={2025 American Control Conference (ACC)},
  pages={2522--2529},
  year={2025},
  organization={IEEE}
}

@book{baier2008principles,
  title={Principles of model checking},
  author={Baier, Christel and Katoen, Joost-Pieter},
  year={2008},
  publisher={MIT press}
}

@inproceedings{maler2004monitoring,
  title={Monitoring temporal properties of continuous signals},
  author={Maler, Oded and Nickovic, Dejan},
  booktitle={International symposium on formal techniques in real-time and fault-tolerant systems},
  pages={152--166},
  year={2004},
  organization={Springer}
}

@book{lavalle2006planning,
  title={Planning algorithms},
  author={LaValle, Steven M},
  year={2006},
  publisher={Cambridge University Press}
}

@book{latombe2012robot,
  title={Robot motion planning},
  author={Latombe, Jean-Claude},
  year={2012},
  publisher={Springer Science \& Business Media}
}

@article{bellman1962dynamic,
  title={Dynamic programming treatment of the travelling salesman problem},
  author={Bellman, Richard},
  journal={Journal of the ACM (JACM)},
  volume={9},
  number={1},
  pages={61--63},
  year={1962},
  publisher={ACM New York, NY, USA}
}

@article{held1962dynamic,
  title={A dynamic programming approach to sequencing problems},
  author={Held, Michael and Karp, Richard M},
  journal={Journal of the Society for Industrial and Applied mathematics},
  volume={10},
  number={1},
  pages={196--210},
  year={1962},
  publisher={SIAM}
}

@article{buck1943partition,
  title={Partition of space},
  author={Buck, Robert Creighton},
  journal={The American Mathematical Monthly},
  volume={50},
  number={9},
  pages={541--544},
  year={1943},
  publisher={Taylor \& Francis}
}

@article{geyer2008optimal,
  title={Optimal complexity reduction of polyhedral piecewise affine systems},
  author={Geyer, Tobias and Torrisi, Fabio D and Morari, Manfred},
  journal={Automatica},
  volume={44},
  number={7},
  pages={1728--1740},
  year={2008},
  publisher={Elsevier}
}

@article{avis1996reverse,
  title={Reverse search for enumeration},
  author={Avis, David and Fukuda, Komei},
  journal={Discrete applied mathematics},
  volume={65},
  number={1-3},
  pages={21--46},
  year={1996},
  publisher={Elsevier}
}

@book{buck2015mazes,
  title={Mazes for Programmers: Code Your Own Twisty Little Passages},
  author={Buck, Jamis},
  year={2015},
  publisher={Pragmatic Bookshelf}
}

@manual{mosek,
   author = "MOSEK",
   title = "The MOSEK Python Fusion API manual. Version 11.0.",
   year = 2025,
   url = "https://docs.mosek.com/latest/pythonfusion/index.html"
 }

@misc{drake,
 author = "Russ Tedrake and the Drake Development Team",
 title = "Drake: Model-based design and verification for robotics",
 year = 2019,
 url = "https://drake.mit.edu"
}

@article{garrett2021integrated,
  title={Integrated task and motion planning},
  author={Garrett, Caelan Reed and Chitnis, Rohan and Holladay, Rachel and Kim, Beomjoon and Silver, Tom and Kaelbling, Leslie Pack and Lozano-P{\'e}rez, Tom{\'a}s},
  journal={Annual review of control, robotics, and autonomous systems},
  volume={4},
  number={1},
  pages={265--293},
  year={2021},
  publisher={Annual Reviews}
}

@article{lavalle2001randomized,
  title={Randomized kinodynamic planning},
  author={LaValle, Steven M and Kuffner Jr, James J},
  journal={The international journal of robotics research},
  volume={20},
  number={5},
  pages={378--400},
  year={2001},
  publisher={SAGE Publications}
}

@article{karaman2011sampling,
  title={Sampling-based algorithms for optimal motion planning},
  author={Karaman, Sertac and Frazzoli, Emilio},
  journal={The international journal of robotics research},
  volume={30},
  number={7},
  pages={846--894},
  year={2011},
  publisher={Sage Publications Sage UK: London, England}
}

@article{kavraki2002probabilistic,
  title={Probabilistic roadmaps for path planning in high-dimensional configuration spaces},
  author={Kavraki, Lydia E and Svestka, Petr and Latombe, J-C and Overmars, Mark H},
  journal={IEEE transactions on Robotics and Automation},
  volume={12},
  number={4},
  pages={566--580},
  year={2002},
  publisher={IEEE}
}

@book{canny1988complexity,
  title={The complexity of robot motion planning},
  author={Canny, John},
  year={1988},
  publisher={MIT press}
}

@inproceedings{ratliff2009chomp,
  title={CHOMP: Gradient optimization techniques for efficient motion planning},
  author={Ratliff, Nathan and Zucker, Matt and Bagnell, J Andrew and Srinivasa, Siddhartha},
  booktitle={2009 IEEE international conference on robotics and automation},
  pages={489--494},
  year={2009},
  organization={IEEE}
}

@article{zucker2013chomp,
  title={Chomp: Covariant hamiltonian optimization for motion planning},
  author={Zucker, Matt and Ratliff, Nathan and Dragan, Anca D and Pivtoraiko, Mihail and Klingensmith, Matthew and Dellin, Christopher M and Bagnell, J Andrew and Srinivasa, Siddhartha S},
  journal={The International journal of robotics research},
  volume={32},
  number={9-10},
  pages={1164--1193},
  year={2013},
  publisher={SAGE Publications Sage UK: London, England}
}

@inproceedings{schulman2013finding,
  title={Finding locally optimal, collision-free trajectories with sequential convex optimization.},
  author={Schulman, John and Ho, Jonathan and Lee, Alex X and Awwal, Ibrahim and Bradlow, Henry and Abbeel, Pieter},
  booktitle={Robotics: science and systems},
  volume={9},
  number={1},
  pages={1--10},
  year={2013},
  organization={Berlin, Germany}
}

@book{betts2010practical,
  title={Practical methods for optimal control and estimation using nonlinear programming},
  author={Betts, John T},
  year={2010},
  publisher={SIAM}
}

@article{natarajan2024implicit,
  title={Implicit graph search for planning on graphs of convex sets},
  author={Natarajan, Ramkumar and Liu, Chaoqi and Choset, Howie and Likhachev, Maxim},
  journal={arXiv preprint arXiv:2410.08909},
  year={2024}
}

@article{morozov2024multi,
  title={Multi-query shortest-path problem in graphs of convex sets},
  author={Morozov, Savva and Marcucci, Tobia and Amice, Alexandre and Graesdal, Bernhard Paus and Bosworth, Rohan and Parrilo, Pablo A and Tedrake, Russ},
  journal={arXiv preprint arXiv:2409.19543},
  year={2024}
}

@article{graesdal2024towards,
  title={Towards tight convex relaxations for contact-rich manipulation},
  author={Graesdal, Bernhard Paus and Chia, Shao Yuan Chew and Marcucci, Tobia and Morozov, Savva and Amice, Alexandre and Parrilo, Pablo A and Tedrake, Russ},
  journal={arXiv preprint arXiv:2402.10312},
  year={2024}
}

@article{cohn2025non,
  title={Non-Euclidean motion planning with graphs of geodesically convex sets},
  author={Cohn, Thomas and Petersen, Mark and Simchowitz, Max and Tedrake, Russ},
  journal={The International Journal of Robotics Research},
  volume={44},
  number={10-11},
  pages={1840--1862},
  year={2025},
  publisher={Sage Publications Sage UK: London, England}
}

@article{luna2026augmented,
  title={Augmented Graphs of Convex Sets and the Traveling Salesman Problem},
  author={Luna, Gael and Summers, Tyler},
  journal={arXiv preprint arXiv:2604.06406},
  year={2026}
}

@inproceedings{philip2024mixed,
  title={A mixed-integer conic program for the moving-target traveling salesman problem based on a graph of convex sets},
  author={Philip, Allen George and Ren, Zhongqiang and Rathinam, Sivakumar and Choset, Howie},
  booktitle={2024 IEEE/RSJ International Conference on Intelligent Robots and Systems (IROS)},
  pages={8847--8853},
  year={2024},
  organization={IEEE}
}

@article{fainekos2009temporal,
  title={Temporal logic motion planning for dynamic robots},
  author={Fainekos, Georgios E and Girard, Antoine and Kress-Gazit, Hadas and Pappas, George J},
  journal={Automatica},
  volume={45},
  number={2},
  pages={343--352},
  year={2009},
  publisher={Elsevier}
}

@article{kress2009temporal,
  title={Temporal-logic-based reactive mission and motion planning},
  author={Kress-Gazit, Hadas and Fainekos, Georgios E and Pappas, George J},
  journal={IEEE transactions on robotics},
  volume={25},
  number={6},
  pages={1370--1381},
  year={2009},
  publisher={IEEE}
}

@article{lahijanian2011temporal,
  title={Temporal logic motion planning and control with probabilistic satisfaction guarantees},
  author={Lahijanian, Morteza and Andersson, Sean B and Belta, Calin},
  journal={IEEE Transactions on Robotics},
  volume={28},
  number={2},
  pages={396--409},
  year={2011},
  publisher={IEEE}
}

@article{plaku2015motion,
  title={Motion planning with temporal-logic specifications: Progress and challenges},
  author={Plaku, Erion and Karaman, Sertac},
  journal={AI communications},
  volume={29},
  number={1},
  pages={151--162},
  year={2015},
  publisher={SAGE Publications Sage UK: London, England}
}

@article{lindemann2018control,
  title={Control barrier functions for signal temporal logic tasks},
  author={Lindemann, Lars and Dimarogonas, Dimos V},
  journal={IEEE control systems letters},
  volume={3},
  number={1},
  pages={96--101},
  year={2018},
  publisher={IEEE}
}

@inproceedings{wolff2014optimization,
  title={Optimization-based trajectory generation with linear temporal logic specifications},
  author={Wolff, Eric M and Topcu, Ufuk and Murray, Richard M},
  booktitle={2014 IEEE International Conference on Robotics and Automation (ICRA)},
  pages={5319--5325},
  year={2014},
  organization={IEEE}
}

@inproceedings{shoukry2017linear,
  title={Linear temporal logic motion planning for teams of underactuated robots using satisfiability modulo convex programming},
  author={Shoukry, Yasser and Nuzzo, Pierluigi and Balkan, Ayca and Saha, Indranil and Sangiovanni-Vincentelli, Alberto L and Seshia, Sanjit A and Pappas, George J and Tabuada, Paulo},
  booktitle={2017 IEEE 56th annual conference on decision and control (CDC)},
  pages={1132--1137},
  year={2017},
  organization={IEEE}
}

@inproceedings{donze2010robust,
  title={Robust satisfaction of temporal logic over real-valued signals},
  author={Donz{\'e}, Alexandre and Maler, Oded},
  booktitle={International conference on formal modeling and analysis of timed systems},
  pages={92--106},
  year={2010},
  organization={Springer}
}

@article{lindemann2020barrier,
  title={Barrier function based collaborative control of multiple robots under signal temporal logic tasks},
  author={Lindemann, Lars and Dimarogonas, Dimos V},
  journal={IEEE Transactions on Control of Network Systems},
  volume={7},
  number={4},
  pages={1916--1928},
  year={2020},
  publisher={IEEE}
}

@article{safaoui2020control,
  title={Control design for risk-based signal temporal logic specifications},
  author={Safaoui, Sleiman and Lindemann, Lars and Dimarogonas, Dimos V and Shames, Iman and Summers, Tyler H},
  journal={IEEE Control Systems Letters},
  volume={4},
  number={4},
  pages={1000--1005},
  year={2020},
  publisher={IEEE}
}

@inproceedings{kaelbling2011hierarchical,
  title={Hierarchical task and motion planning in the now},
  author={Kaelbling, Leslie Pack and Lozano-P{\'e}rez, Tom{\'a}s},
  booktitle={2011 IEEE international conference on robotics and automation},
  pages={1470--1477},
  year={2011},
  organization={IEEE}
}

@article{kaelbling2013integrated,
  title={Integrated task and motion planning in belief space},
  author={Kaelbling, Leslie Pack and Lozano-P{\'e}rez, Tom{\'a}s},
  journal={The International Journal of Robotics Research},
  volume={32},
  number={9-10},
  pages={1194--1227},
  year={2013},
  publisher={Sage Publications Sage UK: London, England}
}

@inproceedings{toussaint2015logic,
  title={Logic-Geometric Programming: An Optimization-Based Approach to Combined Task and Motion Planning.},
  author={Toussaint, Marc},
  booktitle={IJCAI},
  pages={1930--1936},
  year={2015}
}

@article{cambon2009hybrid,
  title={A hybrid approach to intricate motion, manipulation and task planning},
  author={Cambon, Stephane and Alami, Rachid and Gravot, Fabien},
  journal={The International Journal of Robotics Research},
  volume={28},
  number={1},
  pages={104--126},
  year={2009},
  publisher={Sage Publications Sage UK: London, England}
}

@incollection{applegate2011traveling,
  title={The traveling salesman problem: a computational study},
  author={Applegate, David L and Bixby, Robert E and Chv{\'a}tal, Va{\v{s}}ek and Cook, William J},
  booktitle={The traveling salesman problem},
  year={2011},
  publisher={Princeton university press}
}

@article{laporte1992vehicle,
  title={The vehicle routing problem: An overview of exact and approximate algorithms},
  author={Laporte, Gilbert},
  journal={European journal of operational research},
  volume={59},
  number={3},
  pages={345--358},
  year={1992},
  publisher={Elsevier}
}

@article{dumitrescu2003approximation,
  title={Approximation algorithms for TSP with neighborhoods in the plane},
  author={Dumitrescu, Adrian and Mitchell, Joseph SB},
  journal={Journal of Algorithms},
  volume={48},
  number={1},
  pages={135--159},
  year={2003},
  publisher={Elsevier}
}

@article{ny2011dubins,
  title={On the Dubins traveling salesman problem},
  author={Ny, Jerome and Feron, Eric and Frazzoli, Emilio},
  journal={IEEE Transactions on Automatic Control},
  volume={57},
  number={1},
  pages={265--270},
  year={2011},
  publisher={IEEE}
}

@article{savla2008traveling,
  title={Traveling salesperson problems for the Dubins vehicle},
  author={Savla, Ketan and Frazzoli, Emilio and Bullo, Francesco},
  journal={IEEE Transactions on Automatic Control},
  volume={53},
  number={6},
  pages={1378--1391},
  year={2008},
  publisher={IEEE}
}

@article{galceran2013survey,
  title={A survey on coverage path planning for robotics},
  author={Galceran, Enric and Carreras, Marc},
  journal={Robotics and Autonomous systems},
  volume={61},
  number={12},
  pages={1258--1276},
  year={2013},
  publisher={Elsevier}
}

@article{guibas1987linear,
  title={Linear-time algorithms for visibility and shortest path problems inside triangulated simple polygons},
  author={Guibas, Leonidas and Hershberger, John and Leven, Daniel and Sharir, Micha and Tarjan, Robert E},
  journal={Algorithmica},
  volume={2},
  number={1},
  pages={209--233},
  year={1987},
  publisher={Springer}
}

\appendices

\section{Proofs of Optimality Lemmas}
\label{app:proofs}

\subsection*{Proof of Lemma~\ref{lem:sound} (Soundness)}

We show that any path $\hat{p}$ in $\Ghat$ from start node $s$ to the merged target node $t$ satisfies
$\bigwedge_{i=1}^{\nK} (K_i \,\mathcal{R}\, \lnot D_i) \wedge \mathcal{F}(\mathcal{T})$.

\textit{Terminal condition.}  The target node $t$ is the merge of all its
copies across all layers and subgraphs.  Since $t$
corresponds to a free-space set containing $q_T$, any path reaching
$t$ satisfies $\mathcal{F}(\mathcal{T})$.

\textit{Key-door precedence.}  We proceed by induction on the layer
index $\ell$.

\textit{Base case} ($\ell = 0$):  In $\Gcal_\emptyset$, no edge is
incident to any door node by construction of $\mathcal{E}_\emptyset$. So door
nodes are unreachable, and $\lnot D_i$ holds throughout the base layer.
The release formula $K_i \,\mathcal{R}\, \lnot D_i$ is trivially
satisfied since $\lnot D_i$ holds globally.

\textit{Inductive step:}  Suppose that in any subgraph $\Gcal_{S'}$
at layer $\ell - 1$, all keys in $S'$ were visited before entry into
$\Gcal_{S'}$.  Consider subgraph $\Gcal_S$ at layer $\ell$ with
$|S| = \ell$.  By Algorithm~\ref{Algo1} (lines~6--8), the
only way to enter $\Gcal_S$ is via a directed edge at some key node
$K_v^{(S)}$ from $K_v^{(S\setminus\{v\})}$ in $\Gcal_{S\setminus\{v\}}$.
By the inductive hypothesis, all keys in $S \setminus \{v\}$ were
visited before entering $\Gcal_{S\setminus\{v\}}$, and $K_v$ is
visited at the transition edge.  Therefore all keys in $S$ have been
visited before any traversal within $\Gcal_S$.  Within $\Gcal_S$, the
edge set $E_S$ reinserts only edges incident to doors in $\mathbf D(S)$.  The
path can pass through $D_i$ only if $K_i \in S$, which the inductive
argument guarantees was collected prior to entry.  This establishes
$K_i \,\mathcal{R}\, \lnot D_i$ for all $i$.

Finally, the conjunction over all $i$ gives $\hat{p} \models
\phi_{\mathrm{base}}$. \hfill$\square$

\subsection*{Proof of Lemma~\ref{lem:complete} (Completeness)}

Let $q^*$ be a feasible trajectory visiting key regions in order
$K_{i_1}, K_{i_2}, \ldots, K_{i_m}$ ($m \leq \nK$).  Define
accumulated key sets $S_0 = \emptyset$ and $S_j = \{K_{i_1}, \ldots,
K_{i_j}\}$ for $j = 1, \ldots, m$.

(i) \textit{Before first key:} $q^*$ moves through free space in
$\Gcal_\emptyset$.  Since $q^*$ satisfies $\varphi_{\mathrm{base}}$,
it cannot enter any door region.  The convex partition ensures the
segment corresponds to a sequence of adjacent free-space vertices
in $\Gcal_\emptyset$.

(ii) \textit{Key transition at $K_{i_j}$:} Upon $q^*$ entering
$K_{i_j}$, path $\hat{p}$ traverses the directed inter-layer edge
from $K_{i_j}^{(S_{j-1})}$ in $\Gcal_{S_{j-1}}$ to
$K_{i_j}^{(S_j)}$ in $\Gcal_{S_j}$.  This transition has zero cost
and is always available by Algorithm~\ref{Algo1}.

(iii) \textit{Within $\Gcal_{S_j}$:} After collecting $S_j$, $q^*$
may pass through doors in $D(S_j)$.  The edge set $E_{S_j}$ includes
all edges incident to these doors, so each passage corresponds to a
valid edge.

(iv) \textit{Termination:} $q^*$ ends at $q_T$ contained in the
merged target node $t$.

If $q^*$ visits the same key region multiple times, subsequent visits
do not change $S_j$ and $\hat{p}$ remains in the current subgraph.
At each stage, the trajectory segment lies within convex sets of
$\Gcal$, so valid vertex sequences exist.  Therefore $\hat{p}$ is a
valid path in $\Ghat$. \hfill$\square$

\subsection*{Proof of Lemma~\ref{lem:cost} (Cost Equivalence)}

\begin{proof}
The inter-layer directed edges connect copies of the same key node, corresponding to identical convex sets, and are assigned zero cost (Algorithm~2, line 7). All remaining edges in $\hat{\mathcal{G}}$ inherit their cost functions from $\mathcal{G}$, which encode the trajectory cost of problem~\eqref{optprob} via the GCS motion planning framework. Therefore:
\[
\text{cost}(\hat{p}) = \sum_{e = (u,v) \in \hat{E}_{\hat{p}}} \ell_e(x_u, x_v) = \alpha \int_0^T \|\dot{q}(t)\|^2\,dt + \beta T,
\]
where the right-hand side is the objective of \eqref{optprob} evaluated on the corresponding trajectory.
\end{proof}

\section{Proofs of Variation Correctness Theorems}
\label{app:variation_proofs}

\subsection*{Variation 1 (Theorem~\ref{thm:until})}

\textit{Soundness.}  Any path reaching $t^*$ must terminate in
$\Gcal_\Kcal$.  Entry into $\Gcal_\Kcal$ requires traversal of
directed inter-layer edges at all $\nK$ key nodes, since each key's
directed edge is the unique way to increment the layer.  Therefore
all keys are physically visited before $t^*$ is reached.  Combined
with the base soundness argument, $\hat{p} \models
\varphi_{\mathbf{U}}$.

\textit{Completeness.}  Let $q^*$ be feasible under
$\varphi_{\mathbf{U}}$.  Since $\varphi_{\mathbf{U}}$ requires all
$\nK$ keys, the trajectory visits every key region and by the base
completeness argument its path passes through directed edges at all
$\nK$ keys, reaching layer $\nK$ and terminating at
$t^{(\Kcal)} = t^*$. \hfill$\square$

\subsection*{Variation 2 (Theorem~\ref{thm:ordered})}

\textit{Soundness.}  The only directed edge into $\Gcal_{\{1,\ldots,\ell\}}$
is at key node $K_\ell$, from the unique predecessor
$\Gcal_{\{1,\ldots,\ell-1\}}$.  Thus keys must be visited in order
$K_1, \ldots, K_{\nK}$.  The door-opening logic follows from base
soundness: $D_\ell$ is accessible only in $\Gcal_{\{1,\ldots,\ell\}}$,
after $K_\ell$ has been collected.

\textit{Completeness.}  Any trajectory feasible under
$\varphi_{\mathrm{ord}}$ visits keys in order $K_1, \ldots, K_{\nK}$
and its path traverses the unique chain of inter-layer edges in order,
reaching $t^*$ in $\Gcal_{\{1,\ldots,\nK\}}$. \hfill$\square$

\subsection*{Variation 3 (Theorem~\ref{thm:tsp})}

The environment has no door nodes.  Construction~\ref{con:tsp} applies
the Variation~1 construction to a graph $\Gcal$ with $V = \{V_C,
V_K\}$.  The Variation~1 proof applies directly with the door
precedence conditions vacuous.  The top-layer subgraph
$\Gcal_\Kcal$ is identical to $\Gcal$ (no door reinsertion needed),
and the target restricted to $\Gcal_\Kcal$ ensures all waypoints are
visited.  Cost equivalence and the bijection with feasible
trajectories follow from Lemmas~\ref{lem:sound}--\ref{lem:cost}.
\hfill$\square$

\subsection*{Variation 4 (Theorem~\ref{thm:disj})}

\textit{Soundness.}  Door $D_i$ becomes accessible in $\Gcal_S$ only
when $\tilde{K}_i \in S$.  Entry into any subgraph with $\tilde{K}_i
\in S$ requires traversal of a directed edge at either $K_i^A$ or
$K_i^B$.  Either visit satisfies the disjunction $(K_i^A \lor K_i^B)$,
establishing the release condition for $D_i$.

\textit{Completeness.}  Any feasible trajectory collecting either
$K_i^A$ or $K_i^B$ corresponds to traversal of the appropriate
directed edge, transitioning to the subgraph with $\tilde{K}_i \in S$.
\hfill$\square$

\subsection*{Variation 5 (Theorem~\ref{thm:multi})}

\textit{Soundness.}  By the modified mapping~\eqref{eq:multi_unlock},
edges incident to $D_i$ are present in $\Gcal_S$ only when
$\mathcal{K}_i \subseteq S$.  By the inductive argument of
Lemma~\ref{lem:sound}, all keys $K_j \in \mathcal{K}_i$ have
been visited before entry into any such subgraph.

\textit{Completeness.}  Any trajectory collecting all keys in
$\mathcal{K}_i$ before passing through $D_i$ has accumulated $S
\supseteq \mathcal{K}_i$, where edges incident to $D_i$ are
available. \hfill$\square$

\subsection*{Variation 6 (Theorem~\ref{thm:onekey})}

\textit{Soundness.}  Edges incident to $D_j$ are present in $\Gcal_S$
only when $K_{\sigma(j)} \in S$.  By the base inductive argument,
$K_{\sigma(j)}$ has been visited before entry into any such subgraph.
A single key visit simultaneously unlocks all doors in
$D(K_{\sigma(j)})$.

\textit{Completeness.}  Any trajectory collecting $K_{\sigma(j)}$ to
gain access to $D_j$ transitions to $\Gcal_S$ with $K_{\sigma(j)}
\in S$, where all doors in $D(K_{\sigma(j)})$ are accessible.
\hfill$\square$

\subsection*{Variation 7 (Theorem~\ref{thm:timed})}

\textit{Soundness.}  Time window constraints on key and door nodes
ensure any path through $K_i^{(S)}$ assigns $\tau_{K_i} \in [a^K_i,
b^K_i]$, and any path through $D_i^{(S)}$ assigns $\tau_{D_i} \in
[a^D_i, b^D_i]$.  Monotonicity constraints $\tau_v \geq \tau_u$
prevent time reversal.  Combined with base soundness for the release
condition, $\varphi_{\mathrm{time}}$ is satisfied.

\textit{Completeness.}  Any feasible trajectory satisfying
$\varphi_{\mathrm{time}}$ has time coordinates at each node within
the associated convex time set.  Convexity of all added constraints
(time windows, monotonicity) is preserved, so the corresponding path
in $\Ghat$ satisfies all constraints by construction.

\textit{Convexity preservation.}  All added constraints are linear
and the GCS convex optimization framework applies without
modification. \hfill$\square$

\subsection*{Variation 8 (Theorem~\ref{thm:conditional})}
\textit{Soundness.} We consider the two conditional cases for each $i$.

  \textit{Case 1: $A_i$ is never visited ($a_i = 0$ throughout).} The path remains in subgraphs with $a_i = 0$, where $D_i \in \mathbf{D}(\mathbf{a}, S)$ for all $S$ by~\eqref{eq:cond_D}. Door $D_i$ is freely passable throughout, and the antecedent $A_i$ of the conditional is never satisfied. The implication $A_i \Rightarrow (K_i \,\mathcal{R}\, \neg D_i)$ holds vacuously.

  \textit{Case 2: $A_i$ is visited at some point.} The trigger transition edge advances the state to $a_i = 1$, after which $D_i \notin \mathbf{D}(\mathbf{a}, S)$ unless $K_i \in S$ by~\eqref{eq:cond_D}. From this point, the path can only enter $D_i$ in subgraphs where $K_i \in S$, which requires the key transition edge at $K_i$ to have been previously traversed. This establishes $K_i \,\mathcal{R}\, \neg D_i$ from the moment $A_i$ is visited onward.

  In both cases, $\varphi_{\mathrm{cond}}$ is satisfied. The terminal condition follows from the base argument.

\textbf{Completeness.} Let $q^*$ be feasible under $\varphi_{\mathrm{cond}}$. For each $i$, either $q^*$ never visits $A_i$ (its path remains in subgraphs with $a_i = 0$, where $D_i$ is freely accessible), or $q^*$ visits $A_i$ and subsequently satisfies the activated precedence. In the latter case, the trigger transition edge is traversed at the moment $A_i$ is first visited, and the subsequent key-door logic follows the base completeness argument within the activated subgraph. Therefore $q^*$ corresponds to a valid path in $\hat{\mathcal{G}}$.

\section{Maze Benchmark Generator}
\label{app:maze_generator}

The key-door maze benchmark generator creates environments in three
stages.

\subsubsection*{Stage 1: Maze generation}
A perfect maze (a tree with a unique path between any two cells) is
generated using Eller's algorithm~\cite{buck2015mazes}, which processes the
maze one row at a time, requiring at most two rows in memory.  For
natural numbers $r$ and $c$, a maze with $2r+1$ rows and $2c+1$
columns is created.  The start cell is randomly selected as the
center or a corner; the target cell is the cell farthest from the
start by breadth-first search.

\subsubsection*{Stage 2: Key and door placement}
Doors and keys are placed in \emph{batches}, where a batch is an
integer composition of the total number of key-door pairs specifying
how many keys are accessible at each stage.  For example, with 10
key-door pairs, batch $(1,1,2,1,1,1,1,2)$ creates a narrow augmented
GCS (max width 2) and batch $(4,6)$ creates a wide one (max width 6).

Doors are placed first along a path from the target toward the start,
in hallways rather than intersections.  Keys are then placed in dead
ends and corners, beginning at the cell farthest from the start.  The
generator ensures the number of keys equals the number of doors, with
automatic batch-size reduction if placement fails.

\subsubsection*{Stage 3: Optional wall modification}
Walls may be randomly added or removed to adjust augmented GCS width
and potentially create infeasibility.  Removing walls increases
accessible keys per layer; adding walls may block access to keys,
doors, or the goal.

\subsubsection*{Graph creation}
Since cells are on a grid, connectivity is determined directly without
Algorithm~\ref{alg:partition}.  Adjacent cells are merged horizontally
then vertically to reduce graph size.  Cell vertices are labeled
\texttt{s}, \texttt{t}, \texttt{d\#}, \texttt{k\#}, or \texttt{c\#}
for start, target, door, key, or free-space cell, along with an edge
list.

\end{document}